\documentclass[letterpaper,USenglish,11pt]{article}
\usepackage[margin = 1in]{geometry}
\usepackage[utf8]{inputenc}

\usepackage{amsmath}
\usepackage{amsthm}
\usepackage{amsfonts}
\usepackage{amssymb}
\usepackage{graphicx}
\usepackage{wrapfig}
\usepackage{sectsty}
\usepackage{abstract}
\usepackage[labelfont={sf,bf},textfont={sl}]{caption}

\usepackage{algorithm}
\usepackage{algpseudocode}

\usepackage[dvipsnames]{xcolor}
\graphicspath{{./figures/}}
\usepackage{xspace}
\usepackage{hyperref}

\usepackage{thm-restate}

\newtheoremstyle{sfname}%
{}{}%
{\slshape}{}%
{\bfseries\sffamily}{.}%
{ }%
{}

\theoremstyle{sfname}
\newtheorem{theorem}{Theorem}
\newtheorem{lemma}[theorem]{Lemma}
\newtheorem{definition}{Definition}

\newtheorem{observation}{Observation}

\newcommand{\mypar}[1]{\smallskip\noindent{\sffamily\bfseries #1}}

\allsectionsfont{\normalfont\sffamily\bfseries}

\newcommand{\etal}{\textit{et al.}\xspace}

\newcommand{\from}{\colon\xspace}
\newcommand{\eps}{\varepsilon\xspace}
\newcommand{\R}{\mathbb{R}\xspace}
\newcommand{\N}{\mathbb{N}\xspace}

\newcommand{\D}{\mathcal{D}\xspace}
\newcommand{\F}{\mathcal{F}\xspace}

\newcommand{\CP}{\mathit{CP}\xspace}
\newcommand{\BP}{\mathit{BP}\xspace}
\newcommand{\CQ}{\mathit{CQ}\xspace}
\newcommand{\BQ}{\mathit{BQ}\xspace}
\renewcommand{\Im}{\mathrm{Im}\xspace}

\newcommand{\simpl}{\mathrm{simpl}\xspace}

\newcommand{\yes}{YES\xspace}
\newcommand{\no}{NO\xspace}

\newcommand{\normInf}[1]{\ensuremath{\| #1 \|_\infty}}
\newcommand{\norm}[1]{\ensuremath{\| #1 \|}}

\newcommand{\f}{Fr\'echet\xspace}

\title{\sffamily\bfseries A Subquadratic $\boldsymbol{n^\eps}$-approximation\\ for the Continuous \f Distance}

\author{
Thijs van der Horst\thanks{Department of Information and Computing Sciences, Utrecht University, Netherlands and Department of Mathematics and Computer Science, TU Eindhoven, Netherlands; \href{mailto:t.w.j.vanderhorst@uu.nl}{t.w.j.vanderhorst@uu.nl}} \and
Marc van Kreveld\thanks{Department of Information and Computing Sciences, Utrecht University, Netherlands; \href{mailto:m.j.vankreveld@uu.nl}{m.j.vankreveld@uu.nl}} \and
Tim Ophelders\thanks{Department of Information and Computing Sciences, Utrecht University, Netherlands and Department of Mathematics and Computer Science, TU Eindhoven, Netherlands; \href{mailto:t.a.e.ophelders@uu.nl}{t.a.e.ophelders@uu.nl}} \and
Bettina Speckmann\thanks{Department of Mathematics and Computer Science, TU Eindhoven, Netherlands; \href{mailto:b.speckmann@tue.nl}{b.speckmann@tue.nl}}
}
\date{}

\begin{document}

\maketitle

\thispagestyle{empty} 

\begin{abstract}
    \noindent
    The \f distance is a commonly used similarity measure between curves.
    It is known how to compute the continuous \f distance between two polylines with $m$ and $n$ vertices in $\mathbb{R}^d$ in $O(mn (\log \log n)^2)$ time; doing so in strongly subquadratic time is a longstanding open problem.
    Recent conditional lower bounds suggest that it is unlikely that a strongly subquadratic algorithm exists.
    Moreover, it is unlikely that we can approximate the \f distance to within a factor $3$ in strongly subquadratic time, even if $d=1$.
    The best current results establish a tradeoff between approximation quality and running time.
    Specifically, Colombe and Fox (SoCG, 2021) give an $O(\alpha)$-approximate algorithm that runs in $O((n^3 / \alpha^2) \log n)$ time for any $\alpha \in [\sqrt{n}, n]$, assuming $m = n$.
    In this paper, we improve this result with an $O(\alpha)$-approximate algorithm that runs in $O((n + mn / \alpha) \log^3 n)$ time for any $\alpha \in [1, n]$, assuming $m \leq n$ and constant dimension $d$.
\end{abstract}

\newpage

\section{Introduction}

Measuring similarity is a fundamental data analysis task which is used, for example, to cluster data, search in databases, or construct phylogenetic trees. The similarity of two geometric shapes is generally defined via distance measures. Common distance measures are the Hausdorff distance, the area of symmetric difference, and the Fr\'echet distance. While the Hausdorff distance is defined for any two compact subsets of a space, the area of symmetric difference applies only to simple polygons in the plane. The Fr\'echet distance is most easily defined for curves, although generalizations to surfaces exist. Since the Fr\'echet distance respects the order along the two input curves, it is generally the measure of choice to determine the similarity of two curves.

Godau~\cite{godau91continuous_frechet} presented the first polynomial time algorithm for computing the continuous \f distance.
The algorithm runs in $O((mn^2 + m^2 n) \log mn)$ time, and was later improved by Alt and Godau~\cite{alt95continuous_frechet} into an algorithm with running time $O(mn \log mn)$.
Both algorithms rely on an efficient solution for the decision variant of the problem, and then using parametric search~\cite{megiddo83parametric_search} to solve the optimization problem.
Around the same time, the discrete version of the problem was studied by Eiter and Mannila~\cite{eiter94discrete_frechet}, who give an $O(mn)$ time algorithm.

Under the word RAM model of computation, these results have since been improved.
Agarwal~\etal~\cite{agarwal14discrete_frechet} gave an $O(mn \log \log n / \log n)$ time algorithm for the discrete problem, and Buchin~\etal~\cite{buchin17continuous_frechet} improved the complexity bound for the continuous problem to $O(mn (\log \log n)^2)$.
Unfortunately, there is little hope of improving these results further, at least significantly, as Bringmann~\cite{bringmann14hardness} showed that a \emph{strongly subquadratic} (i.e., $(mn)^{1-\Omega(1)}$) time algorithm would refute the \emph{Strong Exponential Time Hypothesis} (SETH).

The main focus has hence been on efficient \emph{approximation algorithms} and on specific classes of curves.
Aronov~\etal~\cite{aranov06frechet_revisited} studied the discrete \f distance between \emph{$\kappa$-bounded} curves and \emph{backbone} curves, which model some ``realistic'' families of curves.
They give a $(1+\eps)$-approximation algorithm running in strongly subquadratic time.
Later on, Driemel~\etal~\cite{driemel12realistic} gave efficient $(1+\eps)$-approximation algorithms for computing the continuous \f distance between various realistic curves.
These families of curves again include $\kappa$-bounded curves, but also \emph{$c$-packed} curves.
An improved algorithm for $c$-packed curves, which matches conditional lower bounds, was later given by Bringmann and K\"unnemann~\cite{bringmann17improved_cpacked}.
Gudmundsson~\etal~\cite{gudmundsson19long} give a $\sqrt{d}$-approximate algorithm for the case where the curves have sufficiently long edges (relative to their \f distance) that runs in linear time.
They give a near-linear time exact algorithm for this case as well.

When dealing with arbitrary curves, SETH again gives conditional lower bounds.
Bringmann~\cite{bringmann14hardness} not only gave a conditional lower bound for the exact problems, but also showed that a $1.001$-approximation algorithm running in strongly subquadratic time would refute the SETH.
This lower bound was later improved by Buchin~\etal~\cite{buchin19seth_says}, who prove that assuming SETH, an approximation factor of $3$ is the best we can hope for, even for curves in one dimension.

For computing the discrete \f distance between arbitrary curves, Bringmann and Mulzer~\cite{bringmann16approximate_discrete} gave an $O(\alpha)$-approximate algorithm with a running time of $O(n^2 / \alpha)$, for any $\alpha \in [1, n / \log n]$.
Chan and Rahmati~\cite{chan18improved_approximation} later gave an improved algorithm that returns an $O(\alpha)$-approximation in $O(n^2 / \alpha^2)$ time, for any $\alpha \in [1, \sqrt{n / \log n}]$.

Bringmann and Mulzer~\cite{bringmann16approximate_discrete} also gave an approximation algorithm for the continuous problem.
The running time of this algorithm is linear, but the approximation ratio is exponential (i.e., $2^{\Theta(n)}$).
This bound was improved only recently, by Colombe and Fox~\cite{colombe21continuous_frechet}.
They give the best result of prior work: an $O(\alpha)$-approximate algorithm that runs in $O((n^3 / \alpha^2) \log n)$ time, for any $\alpha \in [\sqrt{n}, n]$.

\mypar{Results.} We improve upon the result of Colombe and Fox~\cite{colombe21continuous_frechet}.
We give an $O(\alpha)$-approximate algorithm that runs in $O((n + mn / \alpha) \log^3 n)$ time, for any $\alpha \in [1, n]$, if the dimension $d$ of the input curves is constant.

After introducing the necessary preliminaries in Section~\ref{subsec:prelim}, we present an outline of both our algorithm and the paper in Section~\ref{subsec:outline}. Section~\ref{sec:algorithm} then presents our algorithm in detail; the technical Sections~\ref{sec:piecewise_monotone} and~\ref{sec:computing_exit_sets} contain further details for the two major subroutines used by our algorithm.

\subsection{Preliminaries}\label{subsec:prelim}

\mypar{Polygonal curves.}
A (polygonal) \emph{curve} is a piecewise-linear function $P \from [0, 1] \to \R^d$, connecting a sequence $p_1, \dots, p_m$ of $d$-dimensional points, which we refer to as \emph{vertices}.
In this work, we consider curves where the dimension $d$ is constant only.
We will parameterize the domain $[0,1]$ such that the points $p_i$ are uniformly spaced over the domain (i.e. $P((i-1) / (m-1)) = p_i$).
Note that the above definition allows for consecutive vertices to share the same position.
The linear interpolation between $p_i$ and $p_{i+1}$, whose image is equal to the directed line segment $\overline{p_i p_{i+1}}$, is called an \emph{edge}.
We denote by $P[x_1, x_2]$ the subcurve of $P$ over the domain $[x_1, x_2]$.
Similarly, we denote by $P(x_1, x_2)$, the subcurve of $P$ over the open domain $(x_1, x_2)$.
We proceed to define concepts like concatenation and matchings for curves over a closed domain; their adaptations to curves over open domains can be made in the natural way.

For a point $p \in \R^d$, we write $p^\ell$ to denote its $\ell^{\mathit{th}}$ coordinate.
We extend this notation to curves, denoting by $P^\ell \from [0, 1] \to \R$ the curve where $P^\ell(x) = P(x)^\ell$.
Let $P, Q$ be two polygonal curves with vertices $p_1, \dots, p_m$ and $q_1, \dots, q_n$, respectively.
If $p_m = q_1$, we denote by $P \circ Q$ the curve connecting $p_1, \dots, p_m = q_1, \dots, q_n$.

\mypar{\f distance.}
A \emph{reparameterization} is a monotonically increasing, continuous surjection $f \from [0, 1]$ $\to [0, 1]$.
Two orientation-preserving reparameterizations $f, g$ describe a \emph{matching} $(f, g)$ between two curves $P$ and $Q$, where any point $P(f(t))$ is matched to $Q(g(t))$.
A matching $(f, g)$ between $P$ and $Q$ is said to have \emph{cost}
\[
    \max_t \normInf{ P(f(t)) - Q(g(t)) }.
\]
It is common to use the Euclidean norm $\norm{ P(f(t)) - Q(g(t)) }_2$ instead, but for our purposes it is more convenient to use the $L_\infty$-norm. Since we aim for approximation factors between $O(1)$ and $O(n)$, and the norms are at most a factor $\sqrt{d}$ different, approximations using the one norm imply the same approximations for the other norm, as long as $d$ is considered constant.
We refer to a matching with cost at most $\delta$ as a \emph{$\delta$-matching}.
The \emph{\f distance} between $P$ and $Q$ is the minimum cost over all matchings.

\mypar{Free space diagram and matchings.}
Let the square $\D(P, Q) = [0, 1]^2$ define the \emph{parametric space} of $P \times Q$.
A point $(x, y) \in \D(P, Q)$ corresponds to the points $P(x)$ and $Q(y)$ on the two curves.
For $\delta \geq 0$, a point $(x, y) \in \D(P, Q)$ is \emph{$\delta$-close} if $\normInf{P(x) - Q(y)} \leq \delta$.
The \emph{$\delta$-free space} $\F_{\leq \delta}(P, Q)$ of $P$ and $Q$ is the subset of $\D(P, Q)$ containing all $\delta$-close points.

A point $z_2 = (x_2, y_2) \in \F_{\leq \delta}(P, Q)$ is \emph{$\delta$-reachable} from a point $z_1 = (x_1, y_1)$ if there exists a bimonotone path in $\F_{\leq \delta}(P, Q)$ from $z_1$ to $z_2$.
Points that are $\delta$-reachable from $(0, 0)$ are simply called $\delta$-reachable points.
Alt and Godau~\cite{alt95continuous_frechet} observe that the \f distance between $P[x_1, x_2]$ and $Q[y_1, y_2]$ is at most $\delta$ if and only if there is a bimonotone path in $\F_{\leq \delta}(P, Q)$ from $z_1$ to $z_2$.
We can therefore abuse terminology slightly and refer to a bimonotone path from $z_1$ to $z_2$ as a $\delta$-matching between $P[x_1, x_2]$ and $Q[y_1, y_2]$.

The \emph{column} of $P(x)$ is the set $\{x\} \times [0, 1]$.
A maximal vertical line segment $\{x\} \times [y_1, y_2]$ in $\delta$-free space is called a \emph{$\delta$-passage} in the column of $P(x)$.
Given a value $\alpha \geq 1$, an \emph{$(\alpha, \delta)$-exit set} of a point $(0, y)$ with respect to $P$ and $Q$ is a set $E_\alpha(y) \subseteq \{1\} \times [0, 1]$ that $(1)$ contains all points in the column of $P(1)$ that are  $\delta$-reachable from $(0,y)$, 
and $(2)$ for which all $\delta$-close points in $E_\alpha(y)$ are $\alpha\delta$-reachable from $(0,y)$.
An $(\alpha, \delta)$-exit set for $(0, 0)$ with respect to $P$ and $Q$ is simply called an $(\alpha, \delta)$-exit set with respect to $P$ and $Q$.

\subsection{Algorithmic outline}\label{subsec:outline}

Let $P, Q \from [0, 1] \to \R^d$ be our two input curves with $m$ and $n$ vertices, respectively, and let $\alpha \in [1, n]$ be the chosen parameter for the approximation. 
We describe an algorithm that solves the \emph{$O(\alpha)$-approximate decision problem}.
Here we are given a further parameter $\delta \geq 0$. 
Our algorithm answers \yes if $d_F(P, Q) \leq c \alpha \delta$ for some constant $c$, and \no if $d_F(P, Q) > \delta$.
That is, we either confirm that there is a $c \alpha \delta$-matching between $P$ and $Q$ or we assert that no $\delta$-matching exists.
We use the same procedure as Colombe and Fox~\cite{colombe21continuous_frechet} to turn this approximate decision algorithm into an approximation algorithm for the \f distance (with logarithmic overhead in running time and an arbitrarily small increase in approximation ratio).

\begin{figure}[b]
    \centering
    \includegraphics{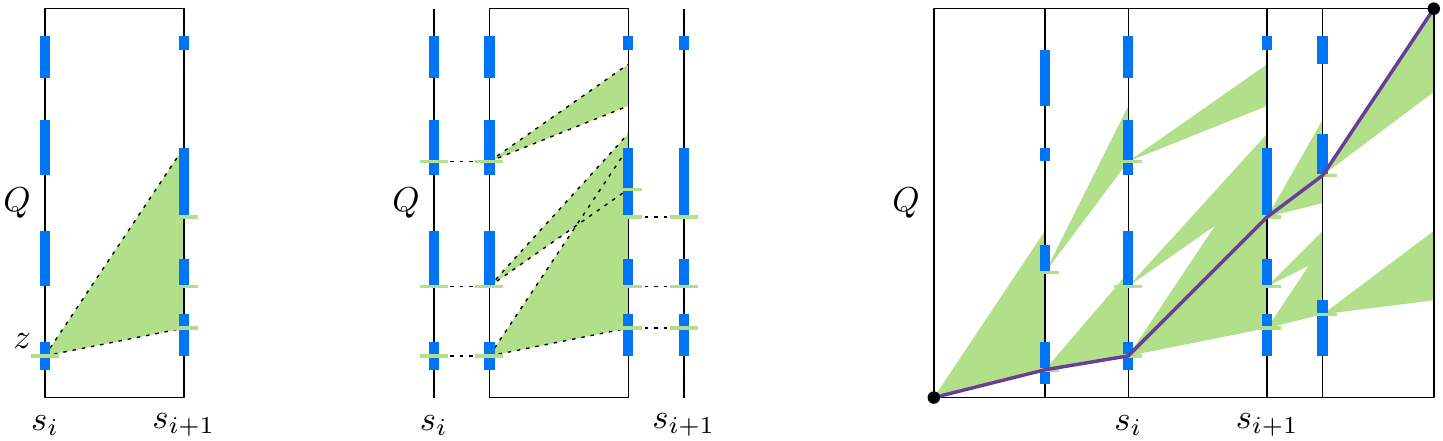}
    \caption{A visual summary of our algorithm; blue segments lie in $\delta$-free space, green regions lie in $O(\alpha \delta)$-free space.
        Vertical bars show sparse signature columns $P(s_i)$ and $P(s_{i+1})$. 
        The blue segments are $\delta$-passages; any matching passes through a blue segment in each column.
        The point $z$ is a point in $S(s_i)$ from where the algorithm explores to find an $(O(\alpha), \delta)$-exit set $E_\alpha(z)$ for $z$; the right side of the green triangle.
        In the middle figure the set $E_\alpha(S(s_i))$ is illustrated.
        The green ticks on the right column form the set $S(s_{i+1})$.
        The last figure shows the parts of free space that are explored by our algorithm, as well as an $O(\alpha \delta)$-matching in purple.
    }
    \label{fig:sparse_columns}
\end{figure}

Recall that a $\delta$-matching between $P$ and $Q$ represents a bimonotone path from $(0, 0)$ to $(1, 1)$ in the $\delta$-free space $\F_{\leq \delta}(P, Q)$.
Our decision algorithm searches for such a path.
The complexity of the free space can be as high as $\Theta(mn)$, so exploring it completely cannot structurally result in a subquadratic algorithm.
We therefore identify so-called \emph{sparse signature columns} in the $\delta$-free space $\F_{\leq \delta}(P, Q)$.
These columns correspond to \emph{signature vertices} that can in some sense be matched to only $O(n / \alpha)$ locations on the other curve.
Signature vertices were originally introduced by Driemel~\etal~\cite{driemel15clustering}; we show in Section~\ref{sec:algorithm} how to generalize their definition for our purposes.
Sparse signature columns have a limited number ($n/\alpha$) of $\delta$-reachable passages. 

Given two sparse signature columns, corresponding to signature vertices $P(s_i)$ and $P(s_{i+1})$, and a set $S(s_i)$ of $\delta$-close points in the column of $P(s_i)$, we compute an $(O(\alpha), \delta)$-exit set $E_\alpha(z)$ for each point $z \in S(s_i)$; Figure~\ref{fig:sparse_columns} shows this for one point $z$ on the left.
This computation is rather technical: it requires using a simplification of a part of $Q$ and replacing $P[s_i,s_{i+1}]$ by a monotone curve.
Once we have done this for all $z\in S(s_i)$, we compute the union $E_\alpha(S(s_i))$ of these exit sets, which contains \emph{all} points in column $P(s_{i+1})$ that are $\delta$-reachable from points in $S(s_i)$, and \emph{only} points that are $O(\alpha \delta)$-reachable from points in $S(s_i)$.
We construct the set $S(s_{i+1})$ by taking the bottom-most point of each connected component of the intersection between $E_\alpha(S(s_i))$ and the $\delta$-passages of column $P(s_{i+1})$; Figure~\ref{fig:sparse_columns} shows this in the middle. Once we have $S(s_{i+1})$, we can continue with the next part of $P$ between the sparse signature vertices $s_{i+1}$ and $s_{i+2}$.
At the end of this iterative construction, if $(1, 1)$ is an element of the final exit set that we compute, then there is a matching between $P$ and $Q$ of cost $O(\alpha \delta)$. Otherwise no matching of cost at most $\delta$ exists.
Figure~\ref{fig:sparse_columns} on the right shows an existing matching in purple.

We show in Section~\ref{sec:algorithm} how to compute the exit sets from any sparse signature column in $O(|P[s_i, s_{i+1}]| (n / \alpha) \log^2 n)$ time.
Together they contain $O(n / \alpha)$ $\delta$-passages, so $S(s_i)$ and $S(s_{i+1})$ have complexity $O(n / \alpha)$.
Constructing $S(s_{i+1})$ therefore takes $O(|P[s_i, s_{i+1}]| (n / \alpha) \log^2 n))$ time. This result depends on the technical Sections~\ref{sec:piecewise_monotone} and~\ref{sec:computing_exit_sets}. There we show that if $\F_{\leq \delta}(P[s_i, s_{i+1}], Q)$ contains no sparse signature columns, then we can compute an $(O(\alpha), \delta)$-exit set for a given point $z$ in time linear in the sizes of $P[s_i, s_{i+1}]$ and $Q$.
Using the simplification data structure described in Section~\ref{sec:algorithm}, we then show how to obtain an improved time bound that depends only logarithmically on the size of $Q$.
For any $z\in S(s_i)$, the exit set $E_\alpha(z)$ is a single vertical line segment, and thus has constant complexity.
We can therefore construct the union $E_\alpha(S(s_i))$ of all exit sets in $O(|S(s_i)| \log |S(s_i)|)$ time. With a single scan, we can now extract $S(s_{i+1})$.

Everything combined, we obtain an $O(\alpha)$-approximate decision algorithm for \f distance that runs in $O((mn / \alpha) \log^2 n + n \log n)$ time, after $O(n \log^3 n)$ preprocessing for the simplification data structure. It is built only once for all decision procedures.

\section{The algorithm}
\label{sec:algorithm}

In this section we present the main algorithm in greater detail.
We are given two curves $P$ and $Q$, with $m$ and $n$ vertices respectively.
To approximate $d_F(P, Q)$, we give an algorithm for the $O(\alpha)$-approximate decision variant of the problem, where we are given parameters $\delta \geq 0$ and $\alpha \in [1, n]$, and we report either \yes or \no, such that when the answer is \yes, we have $d_F(P, Q) \leq c \alpha \delta$ for some constant $c$, and when the answer is \no, we have $d_F(P, Q) > \delta$.
We then use the same procedure as Colombe and Fox~\cite{colombe21continuous_frechet} to turn this approximate decision algorithm into an approximate algorithm for \f distance (with logarithmic overhead in the running time, and an arbitrarily small increase in approximation ratio).
Throughout, we assume $\delta > 0$;
the case $\delta = 0$ can be handled exactly and in linear time, by checking equality of the curves.

\subsection{Sparse vertices}

Recall that a $\delta$-matching between $P$ and $Q$ represents a bimonotone path from $(0, 0)$ to $(1, 1)$ in $\delta$-free space $\F_{\leq \delta}(P, Q)$.
Our decision algorithm looks for such a path in $\delta$-free space.
However, since the complexity of this free space can be as high as $\Theta(mn)$, exploring the entire free space may take $\Omega(mn)$ time.
We therefore identify \emph{sparse signature columns} in free space.
We first repeat the (slightly reworded) definition of the \emph{$\delta$-signature} of a curve in $\R$ by Driemel~\etal~\cite{driemel15clustering}. Then we give an extension to $\R^d$ that we will use.

\begin{definition}[$\delta$-Signature~\cite{driemel15clustering}]
    Let $\delta > 0$ and $P \from [0, 1] \to \R$ be a curve.
    A \emph{$\delta$-signature} of $P$ is a curve $\Sigma \colon [0, 1] \to \R$, whose vertices $P(s_1), \dots, P(s_k)$ are defined by a series of values $0 = s_1 < \dots < s_k = 1$, with the following properties if $k>2$:
    \begin{enumerate}
        \item (non-degeneracy)~
        For all $2 \leq i \leq k-1$: $P(s_i) \notin \overline{P(s_{i-1}) P(s_{i+1})}$.
        \item (approximately direction-preserving)~
        For all $1\leq i\leq k-1$:\\
        If $P(s_i) < P(s_{i+1})$, then for all $s, s' \in [s_i, s_{i+1}]$ with $s < s'$, $P(s) - P(s') \leq 2\delta$.\\
        If $P(s_i) > P(s_{i+1})$, then for all $s, s' \in [s_i, s_{i+1}]$ with $s < s'$, $P(s') - P(s) \leq 2\delta$.
        \item (minimum edge length)\\
        For all $2 \leq i \leq k-2$:  $|P(s_{i+1}) - P(s_i)| > 2\delta$.\\
        $|P(s_{2}) - P(s_1)| > \delta$ and $|P(s_k)-P(s_{k-1})|>\delta$.
        \item (range)~
        For all $s \in [s_i, s_{i+1}]$ with $1\leq i\leq k-1$:\\
        If $2 \leq i \leq k-2$, then $P(s) \in \overline{P(s_i) P(s_{i+1})}$.\\
        If $i = 1$, then $P(s) \in \overline{P(s_1) P(s_{2})} \cup [P(s_1) - \delta, P(s_1) + \delta]$.\\
        If $i = k-1$, then $P(s) \in \overline{P(s_{k-2} P(s_{k-1})} \cup [P(s_{k-1}) - \delta, P(s_{k-1}) + \delta]$.
        \end{enumerate}
        
        If $k = 2$, then property 2 holds and the following version of property 4 for all $s\in[0,1]$:
        
        $P(s) \in \overline{P(0) P(1)} \cup [P(0) - \delta, P(0) + \delta] \cup [P(1) - \delta, P(1) + \delta]$.\medskip
        
    The vertices of $\Sigma$ are called \emph{$\delta$-signature vertices}.
\end{definition}

\begin{definition}[Higher dimensional $\delta$-signature]
    Let $\delta > 0$ and $1 \leq \ell \leq d$, and let $P \colon [0, 1] \to \R^d$ be a curve.
    An \emph{$(\ell, \delta)$-signature} of $P$ is a curve $\Sigma \from [0, 1] \to \R^d$, whose vertices $P(s_1), \dots, P(s_k)$ are defined by a series of values $0 = s_1 < \dots < s_k = 1$ that also define a $\delta$-signature of $P^\ell$.
    The vertices of $\Sigma$ are called \emph{$(\ell, \delta)$-signature vertices}.
\end{definition}

Driemel~\etal~\cite{driemel15clustering} remark that a $\delta$-signature of a curve in $\R$ always exists.
By extension, it follows that an $(\ell, \delta)$-signature of a curve in $\R^d$ always exists as well, for any $1 \leq \ell \leq d$.

Let $1 \leq \ell \leq d$, and let $\Sigma \from [0, 1] \to \R^d$ be an $(\ell, \delta)$-signature of $P$, defined by the values $0 = s_1 < \dots < s_k = 1$.
Let $Q[y_1, y_2]$ be a subcurve of $Q$ that contains at least one vertex of $Q$.
If $Q[y_1, y_2]$ is maximal such that $Q[y_1, y_2]^\ell \subseteq [P(s_i)^\ell - \delta, P(s_i)^\ell + \delta]$, then we call the passage in column $P(s_i)$ corresponding to $Q[y_1, y_2]$ a \emph{candidate $\delta$-passage}.
As we show in the following lemma, any $\delta$-matching between $P$ and $Q$ intersects a candidate $\delta$-passage in each signature vertex column.

\begin{lemma}
\label{lem:matching_signatures}
    Let $1 \leq \ell \leq d$ and $\delta > 0$.
    Let $P, Q \from [0, 1] \to \R^d$ be two curves with $d_F(P, Q) \leq \delta$.
    Any $\delta$-matching between $P$ and $Q$ intersects the column of any signature vertex of $P$ in a candidate $\delta$-passage.
\end{lemma}
\begin{proof}
    Let $(f, g)$ be a $\delta$-matching between $P$ and $Q$.
    Let $s_1, \dots, s_k$ be the values defining an $(\ell, \delta)$-signature of $P$, and let $\sigma_i = P(s_i)$.
    We prove that $(f, g)$ matches every vertex $\sigma_i$ to a point on an edge $e$ of $Q$ that has a vertex with $x^\ell$-coordinate within $[\sigma_i^\ell - \delta, \sigma_i^\ell + \delta]$.
    From this it immediately follows that $(f, g)$ intersects column $P(s_i)$ in a candidate $\delta$-passage, namely the one corresponding to (a curve containing) $e$.
    
    Observe that $(f, g)$ is a $\delta$-matching between $P^\ell$ and $Q^\ell$.
    By definition of $(\ell, \delta)$-signature, $P^\ell(s_i)$ is a signature vertex of $P^\ell$.
    Recall that $P^\ell(s_i) = P(s_i)^\ell = \sigma_i^\ell$ by construction.
    We show that $\sigma_i^\ell$ must be matched to a point on some edge $e^\ell$ of $Q^\ell$ that has a vertex in $[\sigma_i^\ell - \delta, \sigma_i^\ell + \delta]$.
    This means that $\sigma_i$ is matched to a point on an edge $e$ of $Q$ that is projected onto $e^\ell$, and which must therefore have a vertex with $x^\ell$-coordinate in $[\sigma_i^\ell - \delta, \sigma_i^\ell + \delta]$.
    
    Abusing notation slightly, we write $\sigma_i = \sigma_i^\ell$, as well as $P^\ell = P$ and $Q^\ell = Q$.
    We follow the proof of Lemma~3.5 of Driemel~\etal~\cite{driemel15clustering}, who show that if $d_F(P^\ell, Q^\ell) \leq \delta$, then every signature vertex $\sigma_i$ must have a vertex of $Q$ close to it, and these vertices appear on $Q$ in the order of $i$.
    We augment their proof to prove the statement in the theorem.
    
    We say that $Q$ \emph{visits} a vertex $\sigma_i$ of $\Sigma$ when it comes within distance $\delta$ of $\sigma_i$.
    For all $i$, let $Q(y_i)$ be a point matched to $\sigma_i$ by $(f, g)$.
    Note that $Q$ visits $\sigma_{i-1}$, $\sigma_i$ and $\sigma_{i+1}$ in order, as they are also in this order along $P$.
    For $3 \leq i \leq k-2$, the definition of $\delta$-signature gives us that $|\sigma_i - \sigma_{i-1}| > 2\delta$ and $|\sigma_{i+1} - \sigma_i| > 2\delta$.
    Also, it gives us that $\sigma_i \notin \overline{\sigma_{i-1} \sigma_{i+1}}$.
    This means that $Q$ must change direction between visiting $\sigma_{i-1}$ and $\sigma_{i+1}$.
    
    Observe that $Q$ cannot move more than distance $\delta$ away from the edge $\overline{\sigma_{i-1} \sigma_i}$ before it has visited both $\sigma_i$ and $\sigma_{i+1}$, in order.
    This means that $Q[y_{i-1}, y_{i+1}]$ either lies left of $\sigma_i + \delta$ or right of $\sigma_i - \delta$, depending on the positions of $\sigma_{i-1}$, $\sigma_i$ and $\sigma_{i+1}$ relative to each other.
    This shows that any edge of $Q$ that comes close to $\sigma_i$, must have a vertex that lies within distance $\delta$ of $\sigma_i$.
    In particular, this includes any edge that $\sigma_i$ is matched to.
    This completes the proof for $i \in \{3, \dots, k-2\}$.
    
    For $i = 2$, note that by the definition of the Fr\'echet distance, $Q(0)$ must be matched to $\sigma_1 = P(0)$.
    Like before, $Q$ has to visit $\sigma_2$ and $\sigma_3$ in order.
    Either $|Q(0) - \sigma_2| \leq \delta$, in which case the statement is true immediately, or $|Q(0) - \sigma_2| > \delta$.
    In the latter case, because $|\sigma_2 - \sigma_3| > 2\delta$, the proof is analogous to the case where $3 \leq i \leq k-2$.
    The case $i = k-1$ is symmetric.
    
    Finally, by the definition of the Fr\'echet distance, $Q(0)$ and $Q(1)$ must be matched to $\sigma_1 = P(0)$ and $\sigma_k = P(1)$, respectively.
\end{proof}

We say that an $(\ell, \delta)$-signature vertex $\sigma$ is \emph{$\beta$-sparse} if there are $\beta$ candidate $\delta$-passages in its column.
We call $\beta$ the \emph{sparsity} of $\sigma$.
As we show in the following lemma, we can efficiently compute the sparsity of a signature vertex using a range tree (using fractional cascading to improve query times).
We also show how to efficiently compute the candidate $\delta$-passages in the column of $\sigma$.

\begin{lemma}
\label{lem:sparse_data_structure}
    Let $1 \leq \ell \leq d$ and $\delta > 0$.
    Let $P, Q \from [0, 1] \to \R^d$ be two curves, where $Q$ has $n$ vertices.
    With $O(n \log n)$ time preprocessing we can preprocess $Q$ into a data structure of $O(n \log n)$ size that, given an $(\ell, \delta)$-signature vertex $\sigma$ of $P$, computes the sparsity $\beta$ of $\sigma$ in $O(\log n)$ time.
    In $O(\log n + \beta \log \beta)$ additional time, the data structure returns all $\beta$ candidate $\delta$-passages in the column of $\sigma$.
\end{lemma}
\begin{proof}
    Let $q_1, \dots, q_n$ be the $n$ vertices of $Q$, in their order along $Q$.
    Recall that a candidate $\delta$-passage in the column of a signature vertex $\sigma$ corresponds to a maximal subcurve $Q[y_1, y_2]$ of $Q$ that contains at least one vertex of $Q$ and for which $Q[y_1, y_2]^\ell \subseteq R$, where $R = [\sigma^\ell - \delta, \sigma^\ell + \delta]$.
    Note that each such subcurve is uniquely identified by a vertex $q_i$ on $Q$ with $q_i^\ell \in R$ and $q_{i-1} \notin R$ or $i = 1$.
    We can thus count the number of candidate $\delta$-passages in the column of $\sigma$ by counting the number of vertices with the above property.
    Symmetrically, each subcurve corresponding to a candidate $\delta$-passage is uniquely identified by a vertex $q_i$ on $Q$ with $q_i^\ell \in R$ and $q_{i+1} \notin R$ or $i = n$.
    This means that we can report the first and last vertices in the subcurves defining the candidate $\delta$-passages, and construct the passages from them.
    
    We store the vertices of $Q$ in $d$ range trees $T_\ell$ for two-dimensional orthogonal range searching, for $\ell = 1, \dots, d$.
    The tree $T_\ell$ stores the points $(q_{i-1}^\ell, q_i^\ell)$ for $i = 2, \dots, n$, as well as the point $(-\infty, q_1^\ell)$.
    Together with the point $(q_{i-1}^\ell, q_i^\ell)$ (or $(-\infty, q_1^\ell)$), we store the vertex $q_i$ itself, as well as its index $i$.
    This index is used to report the candidate passages.
    Constructing these trees takes $O(n \log n)$ time and they use $O(n \log n)$ space.
    
    To compute the sparsity $\beta$ of an $(\ell, \delta)$-signature vertex $\sigma$, we query $T_\ell$ with the ranges $(-\infty, \sigma^\ell - \delta) \times R$ and $R \times (\sigma^\ell + \delta, \infty)$, count the number of points inside the ranges and add up the reported counts.
    This query returns $\beta$ in $O(\log n)$ time.
    
    We can report all first vertices on the $\beta$ subcurves defining candidate $\delta$-passages in $O(\log n + \beta)$ time using the above range trees.
    With an analogous construction of the trees, we can construct $d$ range trees $T'_\ell$ that report all last vertices on these subcurves in $O(\log n + \beta)$ time.
    The vertices are reported together with their indices.
    Hence we can sort the reported vertices according to their order along $Q$ in $O(\beta \log \beta)$ time.
    We can compute a suitable representation of the candidate $\delta$-passages from these sorted vertices with a single scan of the vertices, taking $O(\beta)$ additional time.
    Thus we can report the $\beta$ candidate $\delta$-passages in the column of $\sigma$ in $O(\log n + \beta \log \beta)$ time.
\end{proof}

\subsection{The algorithm}

We now describe the decision algorithm.
First we identify $(n / \alpha)$-sparse signature vertices.
Let $\Sigma_\ell$ be an $(\ell, \delta)$-signature of $P$, for $\ell = 1, \dots, d$.
Let $0 = s_1 < \dots < s_{k+1} = 1$ be the maximal set of unique values such that for all $2 \leq i \leq k$, $P(s_i)$ is an $(n / \alpha)$-sparse signature vertex.
We split $P$ into the subcurves $P_1 \circ \cdots \circ P_k = P$, such that $P_i = P[s_i, s_{i+1}]$.
In Theorem~\ref{thm:splitting_P}, we show that the subcurves $P_i$ can be constructed in $O((m + n) \log n)$ time.

Theorem~\ref{thm:splitting_P} returns a set of $(\ell, \delta)$-signatures $\Sigma^\ell(P_i)$ of $P_i$ as well, for all $\ell$ and $i$.
These signatures have the following useful property.
For an $(\ell, \delta)$-signature $\Sigma$ with vertices $\sigma_i$, let $R(\Sigma) = \bigcup_i [\sigma_i^\ell - \delta, \sigma_i^\ell + \delta]$ be its \emph{$\delta$-footprint}.
Refer to Figure~\ref{fig:footprint} for an illustration.
The \emph{size} of the $\delta$-footprint is its total length $\lVert R(\Sigma) \rVert$.
The signatures $\Sigma^\ell(P_i)$ returned by the algorithm of Theorem~\ref{thm:splitting_P} all have a $\delta$-footprint of size at most $4\alpha\delta + 4\delta$.
In Section~\ref{sec:computing_exit_sets} we use this property to transform the curves $P_i$ and $Q$ into new curves, whose free space is in a sense an $O(\alpha)$-approximation for the free space between $P_i$ and $Q$.

\begin{figure}
    \centering
    \includegraphics{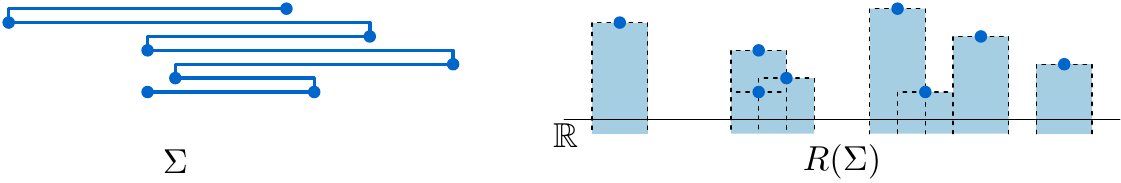}
    \caption{The footprint $R(\Sigma)$ of a signature curve $\Sigma$.
    The blue region on the right side is the footprint of $\Sigma$, shown on the real line.}
    \label{fig:footprint}
\end{figure}

\begin{theorem}
\label{thm:splitting_P}
    Let $\delta > 0$, and let $P, Q \from [0, 1] \to \R^d$ be two curves, with $m$ and $n$ vertices, respectively.
    Let $\alpha \in [1, n]$.
    In $O((m + n) \log n)$ time, we can split $P$ into subcurves $P_1 \circ \dots \circ P_k = P$, such that the vertices joining consecutive subcurves are $(n / \alpha)$-sparse signature vertices.
    For all $\ell$ and $i$, we also compute an $(\ell, \delta)$-signature of $P_i$ with a $\delta$-footprint of size at most $4\alpha\delta + 4\delta$.
\end{theorem}
\begin{proof}
    First we show that the subcurves $P_i$ can be constructed in $O((m + n) \log n)$ time.
    Using the result of Driemel~\etal~\cite{driemel15clustering}, we can compute a $\delta$-signature of $P^\ell$, for any $\ell$, in $O(m)$ time.
    This algorithm is easily extended to return an $(\ell, \delta)$-signature $\Sigma^\ell(P)$ of $P$ in $O(m)$ time.
    Thus, we can compute a set of $(\ell, \delta)$-signatures of $P$, for $\ell = 1, \dots, d$, in $O(m)$ time total.
    
    We construct the data structure of Lemma~\ref{lem:sparse_data_structure} on $Q$.
    This takes $O(n \log n)$ time, and the data structure uses $O(n)$ space.
    For each $\ell$, we query the data structure with each vertex $\sigma$ of $\Sigma^\ell(P)$ to obtain the sparsity of $\sigma$.
    Computing the sparsity of all signature vertices takes $O(m \log n)$ time in total.
    From here, we can construct the subcurves $P_i$ in a single scan of $P$.
    The total time taken to construct the subcurves is $O((m + n) \log n)$.
    
    Because the endpoints of each $P_i$ are signature vertices, the subcurve $\Sigma^\ell(P_i)$ of $\Sigma^\ell(P)$ whose endpoints correspond to those of $P_i^\ell$ is an $(\ell, \delta)$-signature of $P_i$.
    We argue that the $\delta$-footprint of each $\Sigma^\ell(P_i)$ has size at most $4\alpha\delta + 4\delta$, completing the proof.
    
    Consider a signature $\Sigma^\ell(P_i)$ and let $L$ be the size of its $\delta$-footprint.
    Let $I$ be the set indexing the vertices of $\Sigma^\ell(P_i)$.
    We will select a subset $I'$ of $I$ such that $L\leq 2\delta|I'|$, and each vertex of $Q$ is $\delta$-close to at most two signature vertices indexed by $I'$.
    One can construct such a subset $I'$ by initializing $I'$ to be the empty set and sweeping over $R(\Sigma^\ell(P_i))$ from left to right:
    whenever we sweep over a point $p$ that is not close to any signature vertex of $I'$, insert the index of the rightmost signature vertex that is $\delta$-close to $p$ into $I'$.
    This way, each point (and in particular each vertex of $Q$) is $\delta$-close to at most two signature vertices indexed by $I'$.
    Because at most two signature vertices (namely the start and end of $P_i$) are $(n/\alpha)$-sparse, all but two points of $I'$ are $\delta$-close to more than $n/\alpha$ points of $Q$.
    By the pigeonhole principle we have $(|I'|-2) (n/\alpha + 1) \leq 2n$, and hence $|I'|\leq 2\alpha+2$.
    Because $L\leq 2\delta|I'|$, we obtain $L\leq 2\delta|I'|\leq 2\delta(2\alpha+2)=4\alpha\delta+4\delta$.
\end{proof}

We give the main algorithm in the following theorem.
This algorithm uses a procedure for computing exit sets as a black box.

\begin{theorem}
\label{thm:black_box_decision}
    Let $\delta > 0$.
    Given an algorithm that computes a constant-complexity $(O(\alpha), \delta)$-exit set with respect to $P$ and $Q$ in $O(T(m, n))$ time, there is an algorithm for the $O(\alpha)$-approximate decision problem that takes $O((n / \alpha) (\sum_{i} T(m_i, n) + \log n) + (m + n) \log n)$ time, where $m_i$ is the complexity of $P_i$.
\end{theorem}
\begin{proof}
    First compute the subcurves $P_i = P[s_i, s_{i+1}]$ as above with the algorithm of Theorem~\ref{thm:splitting_P}, taking $O((m + n) \log n)$ time.
    We solve the decision problem on $P$ and $Q$ iteratively by computing an $(O(\alpha),\delta)$-exit set at $P(s_{i+1})$ given an $(O(\alpha),\delta)$-exit set at $P(s_i)$.
    For subcurve $P_i = P[s_i, s_{i+1}]$, let $S(s_i)$ be a subset of $O(\alpha \delta)$-reachable points in the column of $P(s_i)$ that are all $\delta$-close.
    Initially, $S(0) = \{(0, 0)\}$ or $S(0) = \emptyset$, depending on whether $(0, 0)$ is $\delta$-close or not.
    Compute an $(O(\alpha), \delta)$-exit set $E_\alpha(z)$ for each point $z \in S(s_i)$ in total $O(|S(s_i)| T(m_i, n))$ time.
    Let $E_\alpha(S(s_i))$ be the union of all computed exit sets.
    Since the resulting exit sets each have constant complexity, their total complexity is $O(|S(s_i)|)$ and hence we can compute $E_\alpha(S(s_i))$ in $O(|S(s_i)| \log |S(s_i)|)$ time.
    
    Report the $O(n / \alpha)$ candidate $\delta$-passages of column $P(s_{i+1})$ in $O(\log n + (n / \alpha) \log (n / \alpha))$ time using the data structure of Lemma~\ref{lem:sparse_data_structure}.
    We extract the bottom-most point of each connected component of the intersection between $E_\alpha(S(s_i))$ and the reported passages.
    This takes $O(|E_\alpha(S(s_i))|) = O(n / \alpha)$ extra time.
    The set $S(s_{i+1})$ is the set containing all these points.
    
    Computing $E_\alpha(S(s_i))$ for all $i$ takes $O(\sum_i |S_i| (T(m_i, n) + \log |S_i|) = O((n / \alpha) (\sum_i T(m_i, n) + \log n)$ time.
    Once $E_\alpha(S(s_k))$ is computed, return \yes if $(1, 1) \in E_\alpha(S(s_k))$, and \no otherwise.
    The claimed running time comes from computing the subcurves $P_i$ and the exit sets $E_\alpha(S(s_i))$.
\end{proof}

We show how to compute an $(O(\alpha), \delta)$-exit set efficiently in Sections~\ref{sec:piecewise_monotone} and~\ref{sec:computing_exit_sets}.
In Section~\ref{sec:piecewise_monotone} we generalize the result of Gudmundsson~\etal~\cite{gudmundsson19long} on curves with a large minimum edge length, to work on \emph{quasi-piecewise $(>4\delta)$-monotone} curves.
This gives a linear-time algorithm for computing a $(3, \delta)$-exit set when one curve is quasi-piecewise $(>4\delta)$-monotone.
In Section~\ref{sec:computing_exit_sets} we approximate two given curves $P$ and $Q$ by two curves $P^*$ and $Q^*$, such that $P^*$ is quasi-piecewise $(>4\delta)$-monotone.
This algorithm implies the following linear-time algorithm for computing exit sets between $P$ and $Q$:

\begin{restatable}{lemma}{lemmaFreeSpaceTraversal}
\label{lem:free_space_traversal}
    Let $\delta > 0$ and $P, Q \from [0, 1] \to \R^d$ be two curves with $m$ and $n$ vertices, respectively.
    Given an $(\ell, \delta)$-signature of $P$ for all $\ell$, whose $\delta$-footprints have a maximum size of $L$, we can compute a constant-complexity $(O(L / \delta), \delta)$-exit set of a point in $O(m + n)$ time.
\end{restatable}

Plugging the above result into the black box of Theorem~\ref{thm:black_box_decision} with $T(m, n) = O(m + n)$, we obtain an $O(L / \delta) = O(\alpha)$-approximate decision algorithm that takes $O((n / \alpha) (\sum_i (m_i + n) + \log n) + (m + n) \log n) = O(mn^2 / \alpha + n \log n)$ time.
As this is super-quadratic, we show in the next section how to lower this time complexity to strongly subquadratic (for $\alpha = \Omega(n^\eps)$).

\subsection{Achieving subquadratic running time}

The dominating factor in the running times of the subproblems comes from the linear dependency on $n$ in Lemma~\ref{lem:free_space_traversal}.
To reduce the running time, we use simplifications $\simpl(Q)_i$ of $Q$ that have complexities roughly equal to those of the subcurves $P_i$.
To obtain a simplification that bounds the approximation ratio sufficiently, we use the data structure of Driemel and Har-Peled~\cite[Section~6.2]{driemel13jaywalking}.
We refer to this data structure as the \emph{simplified subcurve tree}.

The simplified subcurve tree $T(Q)$ can be constructed in $O(n \log^3 n)$ time and uses $O(n \log n)$ space.
Given a subcurve $Q'$ of $Q$ and a parameter $k \in \N$, we can in $O(k \log n)$ time extract a simplification $\simpl_k(Q')$ of $Q'$ with $O(k \log n)$ vertices, as well as a value $\tau(Q', k)$ that upper bounds the distance $d_F(Q', \simpl_k(Q'))$ between $Q'$ and the simplification.
Let $\rho(Q,k)$ be the \f distance from $Q$ to the closest curve with at most $k$ vertices.
Driemel and Har-Peled~\cite{driemel13jaywalking} show that $d_F(Q', \simpl_k(Q'))\leq\tau(Q', k)\leq 11\rho(Q',k)$.\footnote{
Driemel and Har-Peled~\cite{driemel13jaywalking} prove this property only for the case where $Q'$ is stored in a single node of $T(Q)$.
However, it follows from the proof of Lemma~6.8~\cite{driemel13jaywalking} that this property holds for general subcurves of $Q$ as well.
}

Let a subcurve $P_i$ have $m_i$ vertices.
We can directly apply the simplified subcurve tree on $Q$ with parameter $k = m_i$ to obtain a curve with $O(m_i \log n)$ vertices.
However, the induced error $\tau(Q, k)$ will generally be too large.
To obtain meaningful bounds on the induced error, we instead opt to compute a simplification $\simpl_k(Q_z)$ of a subcurve $Q_z = Q[y, y^*]$ that depends on the point $z = (0, y)$ for which we compute an exit set, such that the \f distance between $Q_z$ and $\simpl_k(Q_z)$ is not too large (relative to $\delta$).
To not lose relevant information on the free space between $P_i$ and $Q$, $Q_z$ must contain every subcurve $Q[y, y']$ that has distance $\delta$ to $P_i$, so that the points that are $\delta$-reachable from $z$ will be represented in the free space between $P_i$ and $Q_z$.
In the following theorem, show how to compute the simplification $\simpl_k(Q_z)$ of such a curve in $O(k \log^2 n)$ time, given $T(Q)$.

\begin{figure}
    \centering
    \includegraphics{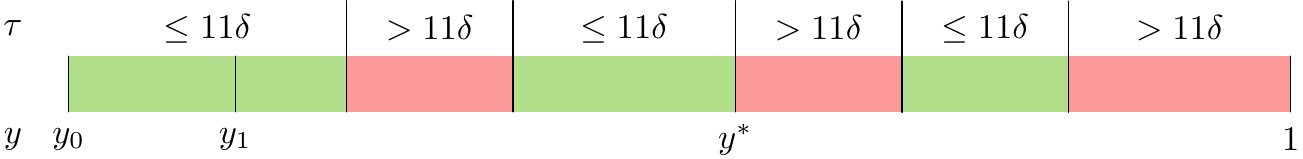}
    \caption{The values of $\tau(Q[y_0, y], k)$ for $y \in [y_0, 1]$.
    The green regions are contained in $T_{11\delta}(y_0, k)$, the red regions are not.
    The supremum $y^*$ of any connected component is at least $y_1$.
    }
    \label{fig:tau_simplifications}
\end{figure}
\begin{theorem}
\label{thm:simplification}
    The simplified subcurve tree $T(Q)$ supports $O(k \log^2 n)$ time queries with parameters $y_0 \in [0, 1]$, $\delta \geq 0$ and $k \in \N$ to retrieve a curve $\simpl_k(Q[y_0, y^*])$ for some $y^* \geq y_0$ that has $\tau(Q[y_0, y^*], k) \leq 11\delta$ and the following property.
    Let $y_1$ be the maximum value such that $\rho(Q[y_0, y_1], k) \leq \delta$, then $y^*\geq y_1$.
\end{theorem}
\begin{proof}
    Let $T_{11\delta}(y_0, k) = \{ y \in [y_0, 1] \mid \tau(Q[y_0, y], k) \leq 11\delta \}$.
    Note that $[y_0, y_1] \subseteq T_{11\delta}(y_0, k)$, since $\tau(Q[y_0, y], k) \leq 11\rho(Q[y_0, y], k) \leq 11\rho(Q[y_0, y_1]) \leq 11\delta$ for any $y \in [y_0, y_1]$.
    Hence, if $y^*$ is the supremum of a connected component of $T_{11\delta}(y_0, k)$, then $y^*\geq y_1$.
    Since $\tau(Q[y_0, y'], k) \leq 11\delta$ for any $y' \in T_{11\delta}(y_0, k)$, we have $d_F(Q[y_0, y'], \simpl_k(Q[y_0, y'])) \leq 11\delta$, so $\simpl_k(Q[y_0, y^*])$ has \f distance at most $11\delta$ to $Q[y_0, y^*]$.
    See Figure~\ref{fig:tau_simplifications}.
    
    We now show how to compute the maximum $y^*$ of some connected component of $T_{11\delta}(y_0, k)$ using a search procedure.
    For $y \geq y_0$, we can query $T(Q)$ with the subcurve $Q[y_0, y]$ and parameter $k$ to obtain a simplification $\simpl_k(Q[y_0, y])$ of $Q[y_0, y]$, as well as a value $\tau := \tau(Q[y_0, y], k)$.
    If $\tau > 11\delta$, then $11\delta < \tau \leq 11 \rho(Q[y_0, y],k)$ and hence $\rho(Q[y_0, y], k) > \delta$, from which it follows that $y > y_1$.
    If instead $\tau \leq 11\delta$, then $d_F(Q[y_0, y], \simpl_K(Q[y_0, y])) \leq \tau \leq 11\delta$, and thus $y \in T_{11\delta}(y_0, k)$.
    In the first case, we search for a lower value of $y$.
    In the latter case, we remember the value $y$ and search for a higher value of $y$ (possibly in a different connected component of $T_{11\delta}(y_0, k)$).
    
    When continuously varying $y$, the value of $\tau(Q[y_0,t],k)$ changes only when $y$ indexes a vertex.
    This follows from the construction of the simplifications, see Theorem~6.9 of Driemel and Har-Peled~\cite{driemel13jaywalking}.
    Hence we can restrict the search to values $y$ that define vertices of $Q$.
    Bisection search on the vertices performs $O(\log n)$ queries on $T(Q)$ and computes an edge of $Q$ whose indexing interval $Y$ contains a valid value $y^*$ in $O(k \log^2 n)$ time.
    Since $\tau(Q[y_0, t], k)$ changes only when $y$ indexes a vertex, we get that $y^*$ indexes a vertex of $Q$ where $\tau(Q[y_0, y^*], k) > 11\delta$ and $\tau(Q[y_0, y^* - \eps], k) \leq 11\delta$ for an arbitrarily small $\eps > 0$.
    We take a small enough $\eps' > 0$ such that $Q(y^* - \eps')$ lies on an edge incident to $Q(y^*)$, and return $\simpl_k(Q[y_0, y^*)) := \simpl_k(Q[y_0, y^* - \eps'], k) \circ Q[y^* - \eps', y^*)$.
\end{proof}

We augment Lemma~\ref{lem:free_space_traversal} to use the above result on the simplified subcurve tree.
Given a $\delta$-close point $z = (s_i, y)$ in the column of $P(s_i)$, we first extract a curve $\simpl_k(Q_z) := \simpl_k(Q[y, y^*])$ from $T(Q)$ according to Theorem~\ref{thm:simplification}.
We then construct $(\ell, 12\delta)$-signatures of $P_i$ that have a small $12\delta$-footprint.
Note that the result of Driemel~\etal~\cite[Lemma~7.1]{driemel15clustering} gives a method to construct $(\ell, 12\delta)$-signatures of $P_i$ that use subsets of the vertices of $\Sigma^\ell(P_i)$.
This process takes $O(m_i \log m_i)$ time per signature, after $O(m_i \log m_i)$ preprocessing~\cite[Lemma~7.5]{driemel15clustering}.
The maximum size of the $12\delta$-footprints of these signatures is at most $12$ times as high as the maximum size of the $\delta$-footprints of the original signatures, since the new signatures use a subset of the vertices.
Thus, their $12\delta$-footprints have size at most $12L = O(\alpha \delta)$.

We apply Lemma~\ref{lem:free_space_traversal} on the curves $P_i$ and $\simpl(Q_z)$, with the point $z' = (0, 0)$ and value $\delta' = 12\delta$.
This gives an $(O(\alpha), \delta')$-exit set $E_\alpha(z')$ for $z'$ in $O(m_i \log n)$ time.
Let $(s_{i+1}, y')$ be a point in the column of $P(s_{i+1})$ that is $\delta$-reachable from $z$.
By Theorem~\ref{thm:simplification} we have we have $y' \leq y^*$ and hence that $d_F(Q[y, y'], Q_z[0, y'']) \leq 11\delta$ for some $y'' \in [0, 1]$.
By the triangle inequality, the point $(1, y'')$ is therefore a $12\delta$-reachable point in $\F_{\leq 12\delta}(P_i, \simpl_k(Q_z))$.
Every point in $\F_{\leq \delta}(P_i, Q)$ that is $\delta$-reachable from $z$ thus corresponds to a point in $\F_{\leq 12\delta}(P_i, \simpl_k(Q_z))$ that is $12\delta$-reachable from $z'$.
Note that it follows from the triangle inequality that every point in $\F_{\leq 12\delta}(P_i, \simpl_k(Q_z))$ that is $12\delta$-reachable from $z'$ corresponds to a point in $\F_{\leq \delta}(P_i, Q)$ that is $23\delta$-reachable from $z$.
Thus we have that $E_\alpha(z')$ corresponds to an $(O(\alpha), \delta)$-exit set for $z$.

Note that the set $E_\alpha(z')$ is represented as a single vertical line segment.
This implies that $E_\alpha(z)$ can be represented by a single vertical line segment as well.
We obtain the following result:

\begin{lemma}
    Let $\delta > 0$ and $P, Q \from [0, 1] \to \R^d$ be two curves with $m$ and $n$ vertices, respectively.
    Given an $(\ell, \delta)$-signature of $P$ for all $\ell$, whose $\delta$-footprints have a maximum size of $L$, we can compute a constant-complexity $(O(L / \delta), \delta)$-exit set of a point in $O(m \log^2 n)$ time.
\end{lemma}

Plugging the above lemma into the black box of Theorem~\ref{thm:black_box_decision} with $T(m, n) = O(m \log^2 n)$ gives our main result: an $O(L / \delta) = O(\alpha)$-approximate decision algorithm that takes $O((m + n) \log n + (n / \alpha) (\sum_i m_i \log^2 n) = O((mn / \alpha) \log^2 n + n \log n)$ time, given $T(Q)$.
We plug this decision algorithm into the optimization procedure of Colombe and Fox~\cite{colombe21continuous_frechet}, to turn it into an algorithm for computing the \f distance.
This does come at the cost of a logarithmic factor in running time, as well as a constant factor in the approximation ratio.
Together with the $O(n \log^3 n)$ preprocessing time to construct $T(Q)$, we obtain the following result:

\begin{theorem}
    Let $P, Q \from [0, 1] \to \R^d$ be two curves, with $m$ and $n$ vertices respectively, and let $\delta > 0$ and $\alpha \in [1, n]$ be parameters.
    Then we can compute an $O(\alpha)$-approximation to $d_F(P, Q)$ in $O((n + mn / \alpha) \log^3 n)$ time.
\end{theorem}

\section{Handling piecewise monotone curves}
\label{sec:piecewise_monotone}

In this section we present an algorithm for computing a $(3,\delta)$-exit set between a quasi-piecewise $(>4\delta)$-monotone curve $P$ (see Definition~\ref{def:monotone_curves}) and an arbitrary curve $Q$ in $O(m+n)$ time.

\begin{restatable}[Piecewise Monotone Curves]{definition}{definitionMonotoneCurves}
\label{def:monotone_curves}
    Let $P \from [0, 1] \to \R^d$ be a curve.
    $P$ is \emph{monotone} if its projection on each of the $d$ coordinate axes is monotonically increasing or decreasing.
    $P$ is \emph{$(>\delta)$-monotone} if moreover $\normInf{P(0) - P(1)} > \delta$.
    $P$ is \emph{piecewise $(>\delta)$-monotone} if it is composed of $(>\delta)$-monotone curves.
    $P$ is \emph{quasi-piecewise $(>\delta)$-monotone} if it is the composition of a piecewise $(>\delta)$-monotone curve with a monotone curve (so the last piece is not necessarily $(>\delta)$-monotone).
\end{restatable}

The following two observations indicate that monotone curves behave similar to line segments.

\begin{observation}
\label{obs:monotone_in_cube}
    Any $\delta$-ball of the $L_\infty$ norm intersects any monotone curve in at most one subcurve.
So if $P$ is monotone, any horizontal line intersects the $\delta$-freespace of $P$ and $Q$ in a convex set.
\end{observation}
\begin{observation}
\label{obs:monotone_moving_closer}
    If $P$ is monotone and $x \leq x'$, then $\normInf{P(x) - P(1)} \geq \normInf{P(x') - P(1)}$.
\end{observation}

We now generalize the definition of the \emph{longest $\delta$-prefix}~\cite{gudmundsson19long}.
\begin{definition}[Longest $\delta$-Prefix]
    The \emph{longest $\delta$-prefix} of a curve $Q$ with respect to a curve $P$, is the largest subcurve $Q[0, y]$ of $Q$ for which $d_F(P, Q[0, y]) \leq \delta$.
\end{definition}
Note that the longest $\delta$-prefix of $Q$ with respect to $P$ exists if and only if $d_F(P,Q[0,y])\leq\delta$ for some $y\in[0,1]$.

\subsection{The structure of an exit set}
For the lemmas in this section, let $P,Q \from [0, 1] \to \R^d$ be two curves, where $P$ is a quasi-piecewise $(>4\delta)$-monotone curve with $k$ pieces $P[x_0,x_1],\dots,P[x_{k-1},x_k]$.
If $d_F(P[0,x_i],Q[0,y])\leq\delta$ for some $y\in[0,1]$, define $y_i$ such that $Q[1,y_i]$ is the longest $\delta$-prefix of $Q$ with respect to $P[0,x_i]$.
    \begin{observation}\label{obs:no_exit}
        If $y_k$ is undefined, then the empty set is a $(1,\delta)$-exit set.
    \end{observation}

We prove in Lemma~\ref{lem:qpm_exit} that a specific $(3,\delta)$-exit set can be represented as an empty set or an interval.
For this, we first show in Lemma~\ref{lem:pm_exit} that if $P$ has the stronger property that it is piecewise $(>4\delta)$-monotone, then a specific $(1,\delta)$-exit set can be represented as an empty set, or an interval.
We then show it Lemmas~\ref{lem:qpm_compute_exit} and \ref{lem:pm_compute_exit} how to compute the respective intervals in linear time.
To compute the exit sets for piecewise $(>4\delta)$-monotone curves, we take advantage of the property (shown below) that any $\delta$-reachable point $(x_{i+1},y)$ is reachable by a monotone path from $(x_i,y_i)$.
    \begin{lemma}\label{lem:pm_exit}
        Let $P$ be piecewise $(>4\delta)$-monotone.
        Assume that $y_k$ is defined and let $y_k^-\in[y_{k-1},y_k]$ be the minimum value such that $y_k^-\geq y_{k-1}$ and $\normInf{P(1)-Q(y_k^-)}\leq\delta$.
        Then
        \begin{enumerate}
            \item[(a)] any $\delta$-reachable point $(1,y)$ has $y\geq y_{k-1}$,
            \item[(b)] for $y\in[y_{k-1},y_k]$, if $(1,y)$ is $\delta$-close, then $(1,y)$ is $\delta$-reachable, and
            \item[(c)] for $y\in[y_k^-,y_k]$, $(1,y)$ is $3\delta$-reachable.
        \end{enumerate}
        Note that properties (a) and (b) simply state that $\{x_k\}\times[y_k^-,y_k]$ is a $(1,\delta)$-exit set.
    \end{lemma}
    \begin{proof}
        If $y_k$ is defined, then so is~$y_{k-1}$.
        We prove the properties (a-c) by induction on the number of pieces $k$.
        The proof for the base case $k=1$ is a special case of $k>0$, so we first prove the case $k>1$ assuming that the lemma holds for curves with fewer than $k$ pieces.
        
        Let $[y^*,y_{k-1}]$ be the exit set obtained by applying the lemma to first $k-1$ pieces of $P$.
        Let $(1,y)$ be any $\delta$-reachable point.
        Because there is no $\delta$-reachable point below $(x_{k-1},y^*)$ and matchings are monotone, we have $y\geq y^*$.
        Assume for a contradiction that $y\in[y^*,y_{k-1}]$.
        By induction, $(x_{k-1},y)$ is $3\delta$-reachable, so $(x_{k-1},y)$ is $3\delta$-close.
        By $(>4\delta)$-monotonicity, we have $\normInf{P(x_{k-1})-P(x_k)}>4\delta$, so by triangle inequality,
        $$\normInf{P(x_k)-Q(y)}\geq\normInf{P(x_{k-1})-P(x_k)}-\normInf{P(x_{k-1})-Q(y)}>4\delta-\delta\geq 3\delta,$$
        contradicting that $(1,y)$ is $\delta$-close, and hence that it is $\delta$-reachable, proving property (a) for $k>1$. 
        For $k=1$, for $y\leq y_{k-1}=y_0$, because $(x_{k-1},y)$ lies on the path that $\delta$-reaches $(x_0,y_0)$, we have that $(x_{k-1},y)$ is $3\delta$-close.
        Now the rest of the proof of (a) for $k>1$ also applies to~$k=1$.
        
        We now prove properties (b) and (c).
        These are trivial for $k=1$, so let $k>1$.
        Because $Q[0,y_k]$ is a $\delta$-prefix with respect to $P[1,x_k]$, there exists a monotone path $\pi$ in $\delta$-free space from $(0,0)$ to $(1,y_k)$.
        For any $y\in[y_{k-1},y_k]$, because $\pi$ does not pass through $\{x_{k-1}\}\times(y_{k-1},1]$, there exists a point $(x,y)$ on $\pi$ with $x\in[x_{k-1},x_k]$.
        
        By Observation~\ref{obs:monotone_in_cube}, if $\normInf{P(1)-Q(y)}\leq\delta$, then the horizontal segment connecting $(x,y)$ to $(x_k,y)$ lies in $\delta$-free space, which together with the part of $\pi$ until $(x,y)$ shows that $(x_k,y)$ is $\delta$-reachable, proving property (b).
        
        For property (c), let $h^-$ be a horizontal segment in $\delta$-free space that ends at $(x_k,y_k^-)$ and starts at a point $(x^-,y_k^-)$ on $\pi$.
        For any $y\in[y_k^-,y_k]$, there exists a point $(x,y)$ on $\pi$ with $x\in[x^-,x_k]$.
        Now $(x,y)$, $(x,y_k^-)$, $(x_k,y_k^-)$ all lie in $\delta$-free space, so by triangle inequality $(x_k,y)$ lies in $3\delta$-free space.
        Hence, by Observation~\ref{obs:monotone_in_cube}, the segment between $(x,y)$ and $(x_k,y)$ lies in $3\delta$-free space, so it is $3\delta$-reachable using the part of $\pi$ up to $(x,y)$ and the horizontal segment, proving property (c).
    \end{proof}

    For the following lemma, we define $y_{k-2}=0$ for the special case that $k=1$.
    \begin{lemma}\label{lem:qpm_exit}
        Let $P$ be quasi-piecewise $(>4\delta)$-monotone.
        Assume that $y_{k-1}$ is defined.
        Let $y_{k-1}^-\in[y_{k-2},y_{k-1}]$ be the minimum value such that $y_{k-1}^-\geq y_{k-2}$ and $\normInf{P(x_{k-1})-Q(y_{k-1}^-)}\leq\delta$.
        Let $$y^+=\begin{cases}
            y_k&\text{if $y_k$ is defined and $y_k>y_{k-1}$,}\\
            y_{k-1}&\text{otherwise.}
        \end{cases}$$
        Then $\{x_k\}\times[y_{k-1}^-,y^+]$ is a $(3,\delta)$-exit set.
    \end{lemma}
    \begin{proof}
        Let $y\in[y_{k-1}^-,y^+]$ and let $(1,y)$ be $\delta$-close.
        It suffices to show that $(1,y)$ is $3\delta$-reachable.
        First consider the case that $y\leq y_{k-1}$, then by Lemma~\ref{lem:pm_exit}, the point $(x_{k-1},y)$ is $3\delta$-reachable (or in the special case that $k=1$, it is even $\delta$-reachable), so if $(x_k,y)$ is $\delta$-close, by Observation~\ref{obs:monotone_in_cube} we can $3\delta$-reach it via the horizontal segment between $(x_{k-1},y)$ and $(x_k,y)$.
        
        Now consider the case that $y>y_{k-1}$, then because $y^+\geq y$, so $y+\neq y_{k-1}$ and hence $y^+=y_k$.
        Let $\pi$ be a monotone path in $\delta$-freespace from $(0,0)$ to $(1,y_k)$.
        Because $\pi$ does not pass through $\{x_{k-1}\}\times(y_{k-1},1]$, $\pi$ contains a point $(x,y)$ with $x\in[x_{k-1},x_k]$.
        So by Observation~\ref{obs:monotone_in_cube}, $(1,y)$ is $\delta$-reachable via $\pi$ and the horizontal segment between $(x,y)$ and $(1,y)$.
    \end{proof}

    \begin{figure}
        \centering
        \includegraphics{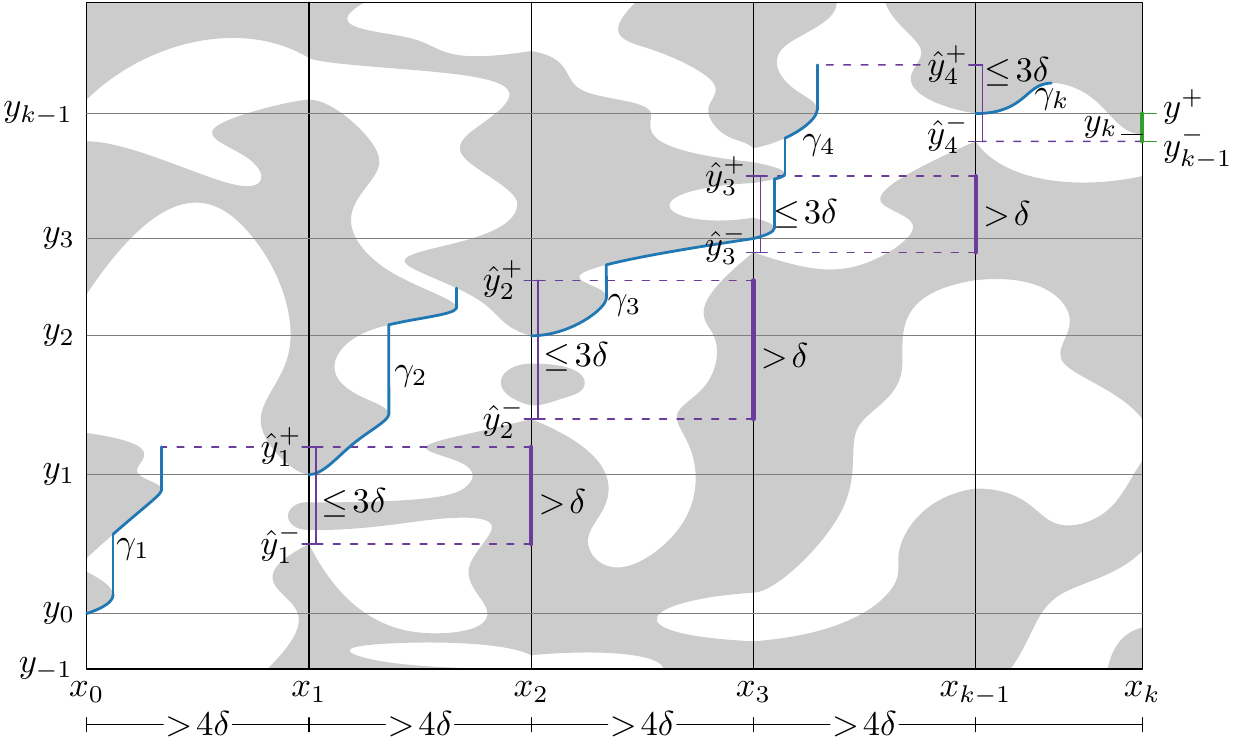}
        \caption{An illustration of the algorithm for a quasi-piecewise monotone curve $P$.}
        \label{fig:greedy_algorithm}
    \end{figure}
    We show how to compute the exit sets of Observation~\ref{obs:no_exit} and Lemmas~\ref{lem:pm_exit} and~\ref{lem:qpm_exit} in linear time.
    The main idea of Lemma~\ref{lem:pm_exit} is to greedily walk through the free space to incrementally compute exit sets of $Q$ with respect to $P[0,x_i]$ for increasing $i$, until either the exit set is empty, or $i=k$.
    The algorithm for quasi-piecewise $(>4\delta)$-monotone curves differs from that of piecewise $(>4\delta)$-monotone curves only in the last step.
    Figure~\ref{fig:greedy_algorithm} illustrates the general algorithm for a quasi-piecewise $(>4\delta)$-monotone curve, where the part of the figure left of $x_{k-1}$ coincides with the algorithm for a piecewise $(>4\delta)$-monotone curve.
    To bound the running time, we linearly bound the running time for iteration $i$ by the involved parts of $P$ and $Q$, and show that any part of $P$ and $Q$ is involved in at most a constant number of iterations.
    
    \begin{lemma}\label{lem:pm_compute_exit}
        Let $P$ be piecewise $(>4\delta)$-monotone.
        We can in $O(m+n)$ time compute the $(1,\delta)$-exit set $\{x_k\}\times[y_k^-,y_k]$ defined in the statement of Lemma~\ref{lem:pm_exit}, or return the empty set if it does not exist.
    \end{lemma}
    \begin{proof}
        Given $[y_{i-1}^-,y_{i-1}]$ as defined by Lemma~\ref{lem:pm_exit} (where $y_0^-=0$), let $\hat y_i^-$ be the minimum value in $[y_{i-1},1]$ such that $\normInf{P(x_i)-Q(\hat y_i^-)}\leq\delta$.
        If no such $\hat y_i^-$ exists, we return the empty set as exit set.
        
        The notation in this proof is illustrated in Figure~\ref{fig:greedy_algorithm}.
        Let $\hat y_i^+\geq\hat y_i^-$ be the maximum value such that $(x_i,y)$ is $3\delta$-close for all $y\in[\hat y_i^-,\hat y_i^+]$.
        By Lemma~\ref{lem:pm_exit}, the vertical segment from $(x_i,\hat y_i^-)$ to $(x_i,\hat y_i^+)$ contains all the $\delta$-reachable points at $x_i$.
        To figure out if any of them are actually reachable, we greedily construct a monotone walk $\gamma_i$ starting at the point $(x_{i-1},y_{i-1})$.
        Let $\gamma_i$ greedily walk upwards through the $\delta$-freespace, advancing rightwards only when needed, and stopping only when it reaches some $(x,y_i^+)$, some $(x_i,y)$, or when we can no longer extend it.
        Suppose that this walk ends at $(x^*,y^*)$.
        Then for all $y\leq y^*$, if $(x_i,y)$ is $\delta$-close, then $(x_i,y)$ is $\delta$-reachable using Observation~\ref{obs:monotone_in_cube}.
        We now show that if $y>y^*$, then $(x_i,y)$ is not $\delta$-reachable.
        Indeed, if it were $\delta$-reachable by a path $\pi$, then $\pi$ contains a point $(x,y^*)$ with $x\in[x^*,x_i]$.
        If $x>x^*$, then Observation~\ref{obs:monotone_in_cube} contradicts that $\gamma_i$ stopped at $(x^*,y^*)$ because it can be extended to the right.
        Similarly if $x=x^*$, then $\gamma_i$ can be extended by following~$\pi$.
        Hence, if $\hat y_i^-\leq y^*$, then $\hat y_i^-=y_i^-$ and $y^*=y_i$, as defined with respect to $P[0,x_i]$ by Lemma~\ref{lem:pm_exit}.
        On the other hand, if $\hat y_i^->y^*$, then we correctly return the empty exit set.
        Iterating this procedure for all $i$ (or until we return the empty exit set for some $i$), we compute the required exit set.
        
        To help analyze the running time, we show that if $y_{i+1}$ is defined, then $y_{i+1}>\hat y_i^+$.
        Indeed, by definition, $\{x_i\}\times[\hat y_i^-,\hat y_i^+]$ lies in $3\delta$-freespace, and because by $(>\delta)$-monotonicity, $P(x_i)$ and $P(x_{i+1})$ differ by $4\delta$, $\{x_{i+1}\}\times[\hat y_i^-,\hat y_i^+]$ does not intersect $\delta$-freespace, so $y_{i+1}>\hat y_i^*$.
        The computation of $\hat y_i^-$ and $\hat y_i^+$ takes time linear in the number of vertices of $Q[y_{i-1},\hat y_i^+]$.
        The computation of $\gamma_i$ takes time linear in the number of vertices of $P$ and $Q$ it passes through.
        $\gamma_i$ passes through vertices of $P[x_{i-1},x_i]$ only, so the contribution of vertices of $P$ over all paths is linear.
        Similarly, because $y_{i+1}>\hat y_i^+$, the computations for $i$ and $i'$ charge different vertices of $Q$ if $|i-i'|\geq 2$, so every vertex is charged at most a constant number of times, bounding the running time by $O(n+m)$.
    \end{proof}
    
    If the last piece of $P$ is not $(>4\delta)$-monotone, then it is possible that $y_k<y_{k-1}$ (as in Figure~\ref{fig:greedy_algorithm}), so the last step of the algorithm for quasi-piecewise $(>4\delta)$-monotone curves is slightly more involved.
    
    \begin{lemma}\label{lem:qpm_compute_exit}
        Let $P$ be quasi-piecewise $(>4\delta)$-monotone.
        We can in $O(m+n)$ time compute the $(3,\delta)$-exit set $\{x_k\}\times[y_{k-1}^-,y^+]$ defined in the statement of Lemma~\ref{lem:qpm_exit}, or return the empty set if it does not exist.
    \end{lemma}
    \begin{proof}
        Using Lemma~\ref{lem:pm_compute_exit}, compute the $(1,\delta)$-exit set $\{x_{k-1}\}\times[y_{k-1}^-,y_{k-1}]$ with respect to the piecewise $(>4\delta)$-monotone curve $P[0,x_{k-1}]$ and $Q$ (or return the empty set if it does not exist).
        It remains to compute $y^+$, for which we recall the definition $$y^+=\begin{cases}
            y_k&\text{if $y_k$ is defined and $y_k>y_{k-1}$,}\\
            y_{k-1}&\text{otherwise.}
        \end{cases}$$
        If $y_k$ is defined and $y_k\geq y_{k-1}$, then the algorithm for computing $\gamma_k$ from Lemma~\ref{lem:pm_compute_exit} will find it, in which case $y^+=y_k$.
        If it is not defined or $y_k<y_{k-1}$, then the algorithm for computing $\gamma_k$ will return the empty set, in which case $y^+=y_{k-1}$.
        The computations of $y_{k-1}$ and $\gamma_k$ both take $O(m+n)$ time.
    \end{proof}
    
    We summarize Lemma~\ref{lem:qpm_compute_exit} in Theorem~\ref{thm:monotone_decision_algorithm}.
    \begin{theorem}
    \label{thm:monotone_decision_algorithm}
        Let $\delta > 0$ and $P, Q \from [0, 1] \to \R^d$ be two curves, where $P$ is quasi-piecewise $(>4\delta)$-monotone.
        Let $m$ and $n$ be the number of vertices of $P$ and $Q$, respectively.
        We can compute a $(3, \delta)$-exit set with respect to $P$ and $Q$ in $O(m + n)$ time.
    \end{theorem}

\section{Constructing piecewise monotone curves with small error}
\label{sec:computing_exit_sets}

Let $P, Q \from [0, 1] \to \R^d$ be two polygonal curves with $m$ and $n$ vertices, respectively.
In this section we show how to approximate $P$ and $Q$ with two curves $P^*$ and $Q^*$, where $P^*$ is quasi-piecewise $(>4\delta)$-monotone.
We can then apply Theorem~\ref{thm:monotone_decision_algorithm} on these curves to compute an exit set with respect to $P$ and $Q$.

The algorithm assumes that we are given a set of $(\ell, \delta)$-signatures $\Sigma^\ell(P)$ of $P$, for $\ell = 1, \dots, d$.
The approximation quality of the algorithm depends on the maximum size of the $\delta$-footprints of these signatures.
Recall from Section~\ref{sec:algorithm} that the $\delta$-footprint of an $(\ell, \delta)$-signature $\Sigma^\ell(P)$ is the region $R(\Sigma^\ell(P)) = \bigcup_i [\sigma_i^\ell - \delta, \sigma_i^\ell + \delta]$.
Throughout this section, we let $L$ be the size of the largest $\delta$-footprint of any of the given signatures.
Note that any signature contains at least one point, and hence we have $L \geq 2\delta$.

We construct $P^*$ and $Q^*$ by considering each dimension $\ell$ separately, and applying point-wise transformations to the $x^\ell$-coordinates of the points.
These transformations induce transformations on the original curves $P$ and $Q$ that are independent of each other.
We approximate $P^\ell$ and $Q^\ell$ with two curves ${P^*}^\ell$ and ${Q^*}^\ell$, which form the projections of $P^*$ and $Q^*$, such that ${P^*}^\ell$ has edges of length greater than $4\delta$ only.
If we have this property for all $\ell$, then $P^*$ is quasi-piecewise $(>4\delta)$-monotone.
Of course, we need that the distance between $P^*$ and $Q^*$ is an approximation for the distance between $P$ and $Q$.
Specifically, if $d_F(P, Q) \leq \delta$, then we require $d_F(P^*, Q^*) \leq \delta$, and if $d_F(P, Q) > cL\delta$ for some $c$, we require $d_F(P^*, Q^*) > \delta$.
Under the $L_\infty$ norm, this is achieved if it holds for all projections of $P^*$ and $Q^*$.

In the remainder of this section, we construct ${P^*}^\ell$ and ${Q^*}^\ell$ for a given dimension $\ell$.
For ease of notation, we write $P = P^\ell$ and $Q = Q^\ell$, and also write $P^* = {P^*}^\ell$ and $Q^* = {Q^*}^\ell$.
The $(\ell, \delta)$-signature $\Sigma^\ell(P)$ induces a $\delta$-signature $\Sigma$ for the curve $P^\ell$.
It follows from the definition of footprints that the $\delta$-footprint of $\Sigma$ is equal to that of $\Sigma^\ell(P)$.
Hence, the $\delta$-footprint of $\Sigma$ has size at most $L$.

\subsection{Splitting up the curves}

We transform $P$ and $Q$ by first partitioning $P$ and $Q$ into subcurves with favorable properties.
In particular, these subcurves are very restricted in how they can be matched to each other, allowing us to apply the transformations that turn the subcurves into line segments.
Since the subcurves have diameter greater than $4\delta$, these line segments will have length greater than $4\delta$ as well.
To partition $P$, after which we partition $Q$, we define a \emph{bridge}.

\begin{definition}[bridge]
    Given a continuous curve $P \from [0, 1] \to \R$ and a subcurve $P'$ ending at point $p$, an interval $I = \overline{p p'}$ is called a \emph{bridge} from $P'$ if $I \cap \Im(P') = \{p\}$.
    The \emph{size} of a bridge $I$ is $\norm{I}$.
\end{definition}

We iteratively split $P$ by identifying bridges in the signature curve $\Sigma$.
Let $\sigma_1, \dots, \sigma_{k'}$ be the vertices of $\Sigma$ and let $\sigma_j = P(s_j)$.
Iterate over the signature vertices $\sigma_j$.
Given $x_i \in [0, 1]$, where initially $x_1 = 0$, do the following.
For a signature vertex $\sigma_j$, let $x \in [x_i, s_j]$ be the largest value for which $P[x_i, x] \in \Im(P[x_i, s_{j-1}])$.
Note that $\overline{P(x) \sigma_j}$ is a bridge from $P[x_i, x]$, due to the range property of the signature.
If $\overline{P(x) \sigma_j}$ has size at most $8\delta$, then continue traversing the signature vertices.
Otherwise, let $x'$ be the largest value such that $P[x_i, x']$ lies within Hausdorff distance $4\delta$ of $P[x_i, x]$.
That is, if $\Im(P[x_i, x]) = [a, b]$, then $\Im(P[x_i, x']) \subseteq [a - 4\delta, b + 4\delta]$.
Define $\CP_i = P[x_i, x']$ and $\BP_i = P(x', s_j)$, and set $x_{i+1} = s_j$.
Once $j = k'$, set $\CP_i = P[x_i, 1]$ and do not define $\BP_i$.
The above results in a decomposition of $P$ into $P = \CP_1 \circ \bigcirc_{2 \leq i \leq k} (\BP_{i-1} \circ \CP_i)$ for some $k$.
For convenience, we write $P_{\leq i} = \bigcirc_{1 \leq j \leq i} (\CP_i \circ \BP_i)$ and $P_{\geq i} = \bigcirc_{i \leq j \leq k} (\CP_i \circ \BP_i)$.

Next we iteratively split up $Q$ based on the decomposition of $P$.
Given $y_i \in [0, 1]$, where initially $y_0 = 0$, do the following.
Let $y$ be the largest value for which $Q[y_i, y]$ is within Hausdorff distance $\delta$ of $\CP_i$.
Let $y' > y$ be the smallest value for which $|\CP_{i+1}(0) - Q(y')| \leq \delta$.
If $y'$ does not exist, set $y' = 1$.
Define $\CQ_i =  Q[y_i, y]$ and $\BQ_i = (y, y')$, and set $y_{i+1} = y'$.
When $y_{i+1} = 1$, set $\CQ_{i+1} = Q(1)$ and do not define $\BQ_{i+1}$.
The above results in a decomposition of $Q$ into $Q = \CQ_1 \circ \bigcirc_{2 \leq i \leq k} (\BQ_{i-1} \circ \CQ_i)$, where $\CQ_k = Q_k$.
For convenience, we write $Q_{\leq i} = \bigcirc_{1 \leq j \leq i} (\CQ_i \circ \BQ_i)$ and $Q_{\geq i} = \bigcirc_{i \leq j \leq k} (\CQ_i \circ \BQ_i)$.
Note that the number of curves $\CP_i$ and $\CQ_i$, as well as $\BP_i$ and $\BQ_i$, are equal.
See Figure~\ref{fig:subcurves} for an illustration of the curves $\CP_i$, $\BP_i$, $\CQ_i$ and $\BQ_i$.

\begin{figure}
    \centering
    \includegraphics{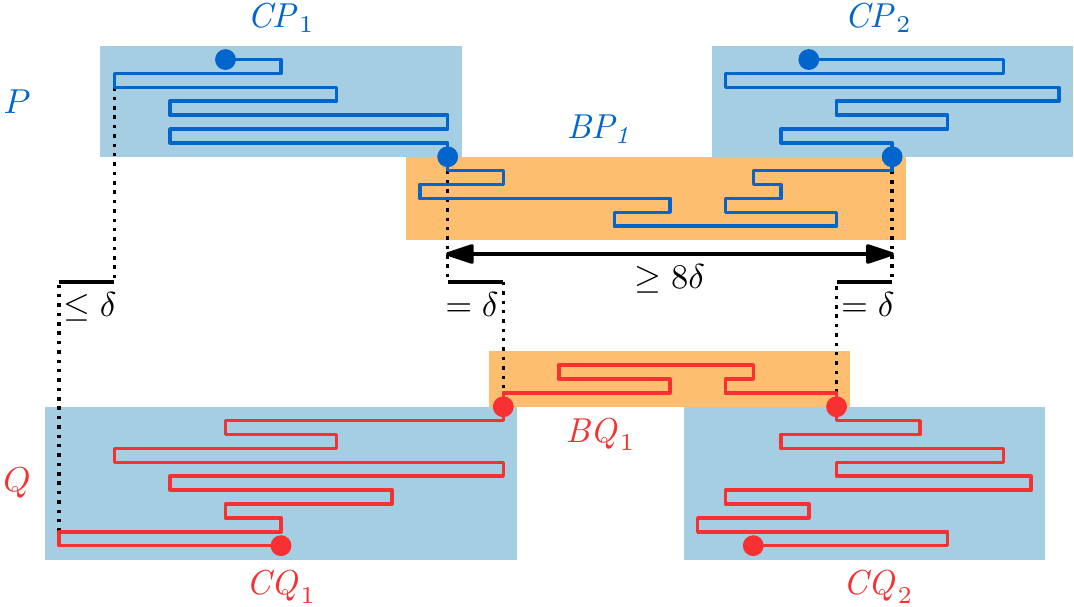}
    \caption{The different subcurves that $P$ (blue) and $Q$ (red) are split into.
    The curves move horizontally in one dimension, but the vertical segments are added for clarity.
    Curves in the blue areas are the compact curves $\CP_i$ and $\CQ_i$, and the curves in the orange areas are the bridge curves $\BP_i$ and $\BQ_i$.
    }
    \label{fig:subcurves}
\end{figure}

We can construct the curves $\CP_i$, $\BP_i$, $\CQ_i$ and $\BQ_i$ in a single scan of $P$ and $Q$, given a signature curve $\Sigma$ of $P$.
We obtain the following lemma:

\begin{lemma}
    Given a signature curve $\Sigma$ of $P$, we can decompose $P$ and $Q$ into the subcurves $\CP_i$, $\BP_i$, $\CQ_i$ and $\BQ_i$ in $O(m + n)$ time.
\end{lemma}

The constructed subcurves have several useful properties.
The \emph{compact} curves $\CP_i$ have small diameter, and the \emph{bridge} curves $\BP_i$ are \emph{$2\delta$-direction-preserving} and have sufficiently large diameter.
We define a direction-preserving curve based on the definition of $\delta$-signatures by Driemel~\etal~\cite{driemel15clustering}, who require $\delta$-signatures to be direction-preserving.

\begin{definition}[$\delta$-Direction-Preserving]
    Let $\delta > 0$.
    A curve $P \from [0, 1] \to \R$ is \emph{$\delta$-direction-preserving} in the positive direction if the following holds.
    Let $0 = x_1 < \dots < x_m = 1$ index the vertices of $P$.
    Then $P(x) - P(x') \leq \delta$ for all $x, x' \in [x_i, x_{i+1}]$ with $x < x'$.
    A $\delta$-direction-preserving curve in the negative direction is defined analogously.
\end{definition}

The following two lemmas prove the properties of the compact and bridge curves.
Lemma~\ref{lem:compact_properties} states that each compact curve has diameter at most $4L$.
Lemma~\ref{lem:bridge_properties} states that the bridge curves are all $2\delta$-direction-preserving, and have a minimum diameter greater than $4\delta$.
The small diameter of the compact curves means that we can essentially ignore the entire subcurve, without incurring too much error in the approximation.
The sufficiently large diameter of the bridge curves allow us to construct a curve $P^*$ that has minimum edge length greater than $4\delta$.

\begin{lemma}
\label{lem:compact_properties}
    Each $\CP_i$ has diameter at most $4L - 2\delta$.
\end{lemma}
\begin{proof}
    We constructed $\CP_i$ by traversing $P_i$ over its signature vertices, and based on the distance between a signature vertex and the current curve, we either expand the curve to contain this signature vertex, or we stop expanding the curve.
    In both cases, the diameter of the curve grows by at most $8\delta$ (in the latter case, we grow the curve slightly before we stop expanding it).
    
    Let $R(\Sigma)$ be the $\delta$-footprint of $\Sigma$.
    Recall that the size of $R(\Sigma)$ is at most $L$.
    As each connected component of $R(\Sigma)$ has size at least $2\delta$, there can not be more than $L/(2\delta)$ connected components in $R(\Sigma)$.
    During construction of $\CP_i$, the curve always ends at a vertex of $\Sigma$ (except possibly when the construction is done).
    We therefore have that each expansion of the curve either keeps the endpoint of the curve in the same connected component of $R(\Sigma)$, or it jumps to another connected component.
    In the first case, the total expansion the curve can do is bounded by the size of $R(\Sigma)$, which is $L$.
    In the second case, the total expansion is bounded by $8\delta$ times the number of connected components of $R(\Sigma)$, which is at most $L/(2\delta)$.
    This gives $\CP_i$ a diameter of at most $4L$.
    Note that $\CP_i$ does not come strictly within distance $\delta$ of the leftmost and rightmost points of $R(\Sigma)$, as there are no signature vertices there.
    Thus we can subtract $2\delta$ from the bound, giving $\CP_i$ a diameter of at most $4L - 2\delta$.
\end{proof}

\begin{lemma}
\label{lem:bridge_properties}
    Each $\BP_i$ is $2\delta$-direction-preserving and has diameter greater than $4\delta$.
\end{lemma}
\begin{proof}
    The curve $\BP_i$ is a subcurve of $P(s_j, s_{j+1})$, where $P(s_j)$ and $P(s_{j+1})$ are two consecutive vertices of $\Sigma$.
    This directly implies that $\BP_i$ is $2\delta$-direction-preserving.
    Because $\BP_i$ is constructed by taking a subcurve $P[x, s_{j+1}]$ of $P[s_j, s_{j+1}]$ that has diameter greater than $8\delta$, of which we set $\BP_i = P[x', s_{j+1}]$ such that $P[x, x']$ has diameter $4\delta$, it follows that the diameter of $\BP_i$ is greater than $4\delta$.
\end{proof}

Transforming two curves $P$ and $Q$ is usually a global process; if nothing is known about what a point $p \in P$ can and cannot be matched to, we cannot move $p$ without potentially having to transform all of $Q$ (and all of $P$ because of this).
The subcurves $\CP_i$, $\BP_i$, $\CQ_i$ and $\BQ_i$ are very restricted in what they can be matched to however, allowing us to make local changes to the curves.
In Theorem~\ref{thm:matching_restrictions} we give the various restrictions on matchings.
We use these restrictions in Section~\ref{subsub:transforming_subcurves} to apply local transformations to $P$ that allow us to construct the transformations based on local properties.
This is key in computing a curve $P^*$ that has long edges.

\begin{theorem}
\label{thm:matching_restrictions}
    Assume that $d_F(P, Q[0, y]) \leq \delta$ for some $y \in [0, 1]$ and let $(f, g)$ be a $\delta$-matching between $P$ and $Q[0, y]$.
    Then $(f, g)$ matches the various subcurves to each other with the following restrictions:
    \begin{itemize}
        \item \label{item:CP_matching} $\CP_i$ is matched to a subcurve of $\CQ_i$, for all $1 \leq i \leq k$.
        \item \label{item:BP_matching} $\BP_i$ is matched to a subcurve of $\CQ_i \circ \BQ_i \circ \CQ_{i+1}$ for all $1 \leq i \leq k-1$.
    \end{itemize}
\end{theorem}
\begin{proof}
    Without loss of generality, we assume that $\BP_i$ and $\BQ_i$ exist.
    We give a proof by induction over $i$.
    First we prove the base case for $\CP_i$.
    By construction, there is an $\eps > 0$ such that the distance between any point on $\CP_1$ and any point on $\BQ_1(0, \eps]$ is greater than $\delta$.
    Since $P(0) = \CP_1(0)$ must be matched to $Q(0) = \CQ_1(0)$, it follows that $\CP_1$ must be matched to a subcurve of $\CQ_1$, proving the base case for $\CP_i$.
    
    We now prove the base case for $\BP_i$.
    Say there is a point $p = \BP_1(x)$ that is matched to $\CQ_2(0)$.
    Assume without loss of generality that $\BP_1$ moves to the right.
    Then because $\BP_1$ is $2\delta$-direction-preserving (Lemma~\ref{lem:bridge_properties}), we have that $\BP_1[x, 1) \subseteq [p - 2\delta, \CP_2(0)) \subseteq[\CP_2(0) - 4\delta, \CP_2(0))$.
    Note that $| \CQ_2(1) - \CP_2(0) | \geq 5\delta$, and therefore that there is an $\eps > 0$ such that all points on $\BQ_2(0, \eps]$ have distance greater than $5\delta$ to $\CP_2(0)$.
    We obtain that no point on $\BQ_2(0, \eps]$ can be matched to points on $\BP_1[x, 1)$.
    It follows that points on $\BP_1$ cannot be matched to points on $\BQ_2$, and thus that $\BP_1$ is matched to a subcurve of $\CQ_1 \circ \BQ_1 \circ \CQ_2$.
    This proves the base case for $\BP_i$.
    
    We now prove the induction step, starting with $\CP_i$.
    Similar to the base case, we have that there is an $\eps > 0$ such that the distance between any point on $\CP_i$ and any point on $\BQ_i(0, \eps]$ is greater than $\delta$.
    It follows that $\CP_i$ must either be matched to a subcurve of $Q_{\leq i-1} \circ \CQ_i$, or a subcurve of $\BQ_i(\eps, 1] \circ Q_{\geq i+1}$.
    By the induction hypothesis, $\BP_{i-1}$ must be matched to a subcurve of $\CQ_{i-1} \circ \BQ_{i-1} \circ \CQ_i$.
    Hence, $\CP_i$ must be matched to a subcurve of $\CQ_{i-1} \circ \BQ_{i-1} \circ \CQ_i$, as otherwise $\BQ_i(0, \eps]$ would not be matched to anything.
    Note that $\CP_i(0)$ does not lie close to any point on $\CQ_{i-1} \circ \BQ_{i-1}$.
    We get that $\CP_i$ must be matched to a subcurve of $\CQ_i$, proving the induction step for $\CP_i$.
    
    We conclude the proof by proving the induction step for $\BP_i$.
    Through similar reasoning as for the base case, it can be shown that if a point on $\BP_i$ is matched to $\CQ_{i+1}(0)$, then no point on $\BP_i$ can be matched to points on $\BQ_{i+1}$.
    By the induction hypothesis, $\CP_i$ is matched to a subcurve of $\CQ_i$.
    This means that either $\BP_i$ is matched to a subcurve of $\CQ_i$, or there is a point on $\BP_i$ that is matched to $\CQ_{i+1}(0)$.
    In both cases, it holds that points on $\BP_i$ cannot be matched to points on $\BQ_{i+1}$, from which it follows that $\BP_i$ is matched to a subcurve of $\CQ_i \circ \BQ_i \circ \CQ_{i+1}$.
    This proves the induction step for $\BP_i$.
\end{proof}

\subsection{Transforming the subcurves}
\label{subsub:transforming_subcurves}

We alter the curves $\CP_i$, $\BP_i$, $\CQ_i$ and $\BQ_i$, so that we can concatenate them into two curves $P^*$ and $Q^*$ with the desired properties.
Recall that we want $P^*$ to have edges that all have length greater than $4\delta$.
Also, we want the \f distance between $P^*$ and $Q^*$ to be approximately equal to the \f distance between $P$ and $Q$, so that we can use $P^*$ and $Q^*$ to compute an exit set with respect to $P$ and $Q$.

We make the assumption that $\BQ_i$ is $4\delta$-direction-preserving.
Observe that if it is not, then no subcurve $Q[0, y]$ of $Q$ exists with $d_F(P, Q[0, y]) \leq \delta$, meaning that the empty set is a $(1, \delta)$-exit set.
Indeed, by Theorem~\ref{thm:matching_restrictions}, any $\delta$-matching matches $\BQ_i$ to a subcurve of $\BP_i$.
Because $\BP_i$ is $2\delta$-direction-preserving (Lemma~\ref{lem:bridge_properties}), it follows that $\BQ_i$ must be $4\delta$-direction-preserving for a subcurve $Q[0, y]$ with $d_F(P, Q[0, y]) \leq \delta$ to exist.

We describe how to construct $P^*$ and $Q^*$.
We call the transformations that we apply to the curves \emph{straightenings}.
As the name implies, it straightens a curve, in the sense that it makes the curve move in only a single direction.
Straightenings are formally defined as follows:

\begin{definition}
    Given a curve $P \from [0, 1] \to \R$ and interval $I = [a, b]$, the \emph{$I$-straightening} of $P$ in the positive direction is the curve $P' \from [0, 1] \to \R$, where 
    \[
        P'(x) = \begin{cases}
            \min\{\max_{x' \leq x} P(x'), b\} & \text{if $P(x) \in I$.} \\
            P(x) & \text{otherwise.}
        \end{cases}
    \]
    The $I$-straightening in the negative direction is defined analogously, with the roles of $\min$ and $\max$, as well as the roles of $a$ and $b$, switched.
\end{definition}
\begin{figure}[bh]
    \centering
    \includegraphics{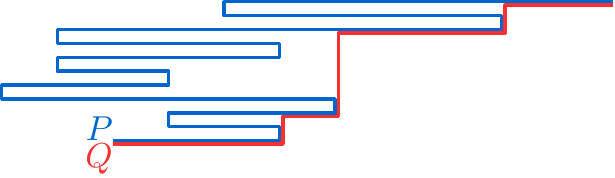}
    \caption{The $\R$-straightening $Q$ (red) of a curve $P$ (blue) in the positive direction.
    The curve $P$ moves horizontally in one dimension, but the vertical segments are added for clarity.
    Points on the blue curve are mapped to points on the red curve $Q$ at the same height.}
    \label{fig:straightening}
\end{figure}
See Figure~\ref{fig:straightening} for an illustration of a straightening.
Note that the straightening of a curve in one dimension yields merely a line segment.
However, due to its point-wise definition, straightenings can be applied to curves in higher dimensions as well, which is required for our case.

The straightening of a curve can be computed in time linear in the number of vertices of the curve by performing the straightening operation over the vertices and connecting the new vertices with line segments.
We transform the curves $\CP_i$, $\BP_i$, $\CQ_i$ and $\BQ_i$ using straightenings.
We iterate over the indices $i = 1, \dots, k$.
Assume that $\CP_{i+1}(0) > \CP_i(1)$.
The other case is handled symmetrically.
Let $\Im(\CP_i) = [a, b]$ and set $I = [a - \delta, \CP_{i+1}(0) - 4\delta]$.
We construct the $I$-straightenings $\CP_i' \circ \BP_i'$ and $\CQ_i' \circ \BQ_i'$ of $\CP_i \circ \BP_i$ and $\CQ_i \circ \BQ_i$, respectively, in the positive direction.
Also, if $i \leq k-1$, let $I' = [\CP_{i+1}(0) - 4\delta, \CP_{i+1}(0)]$.
If $\CP_{i+1}(1) > \CP_{i+1}(0)$, then we construct the $I'$-straightenings $\BP_i''$ and $\BQ_i''$ of $\BP_i'$ and $\BQ_i'$, respectively, in the positive direction.
Else we construct the $I'$-straightenings in the negative direction.

We choose the straightening intervals in such a way that both the compact curves $\CP_i$ and bridge curves $\BP_i$ are straightened into line segments.
(For $\CP_i$, this requires only the $I$-straightening, but for $\BP_i$ this requires both the $I$- and $I'$-straightenings.)
The properties of $\CP_i$ and $\BP_i$ give all resulting line segments lengths greater than $4\delta$ (except for the last).
Hence, the curve $P^* = \CP_1' \circ \bigcirc_{2 \leq i \leq k} (\BP_{i-1}'' \circ \CP_i')$ has long edges only, except for the last.
This is proved in the following lemma.

\begin{lemma}
\label{lem:min_edge_length}
    All edges of $P^*$, except possibly the last edge, have length greater than $4\delta$.
\end{lemma}
\begin{proof}
    First we show that $\CP_i' = \overline{\CP_i(0) \CP_i(1)}$ for all $i$.
    Let $\Im(\CP_i) = [a, b]$, and assume without loss of generality that $\CP_{i+1}(0) > \CP_i(1)$.
    Note that this also means that $\CP_i(1) \geq \CP_i(x)$ for all $x \in [0, 1]$.
    Let $I = [a - \delta, \CP_{i+1}(0) - 4\delta]$.
    Then $\CP_i'$ is the $I$-straightening of $\CP_i$.
    Since $\CP_i \subseteq I$, we have that $\CP_i'(x) = \max_{x' \leq x} \CP_i(x)$ for all $x \in [0, 1]$, and thus $\CP_i'(0) \leq \CP_i'(x) \leq \CP_i'(1)$.
    Note that $\CP_i'(0) = \CP_i(0)$ and $\CP_i'(1) = \CP_i(1)$.
    It follows that $\CP_i' = \overline{\CP_i(0) \CP_i(1)}$.
    
    Similar reasoning as above shows that $\BP_i' \cap I$ forms a single line segment, oriented in the same direction as $\CP_i'$.
    Because $\BP_i' \setminus I \subseteq I'$ for $I' = [\CP_{i+1}(0) - 4\delta, \CP_{i+1}(0)]$, similar reasoning as above again shows that either $\BP_i'' = \BP_i'$, or $\BP_i'' \setminus \BP_i'$ is a single line segment.
    In both cases, $\BP_i''$ is a single line segment.
    Thus, $\CP_i' \circ \BP_i''$ is a single line segment.
    
    Because $P^* = \CP_1' \circ \bigcirc_{2 \leq i \leq k} (\BP_{i-1}'' \circ \CP_i')$, it follows that if every $\CP_i' \circ \BP_i''$ has length greater than $4\delta$, except for when $i = k$, then all but the last edge of $P^*$ has length greater than $4\delta$.
    That $\CP_i' \circ \BP_i''$ has length greater than $4\delta$ follows from the fact that $\CP_i' = \overline{\CP_i(0) \CP_i(1)}$, and because we constructed $\CP_i$ such that $|\CP_i(0) - \CP_i(1)| \geq 4\delta$.
\end{proof}

Let $Q^* = \CQ_1' \circ \bigcirc_{2 \leq i \leq k} (\BQ_{i-1}'' \circ \CQ_i')$.
In Lemma~\ref{lem:exit_sets} we show that a $(c, \delta)$-exit set with respect to $P^*$ and $Q^*$, with $c = O(1)$, is also a $(4L / \delta, \delta)$-exit set with respect to $P$ and $Q$.
This theorem follows from the following two lemmas, which show that the \f distance between $P$ and subcurves of $Q$ is approximately preserved between $P^*$ and subcurves of $Q^*$.

\begin{lemma}
\label{lem:distances_preserved}
    If $d_F(P, Q[0, y^*]) \leq \delta$ for some $y^* \in [0, 1]$, then $d_F(P^*, Q^*[0, y^*]) \leq \delta$.
\end{lemma}
\begin{proof}
    We show that the matching $(f, g)$ between $P^*$ and $Q^*[0, y^*]$ induced by (the same) $\delta$-matching $(f, g)$ between $P$ and $Q[0, y^*]$ has cost at most $\delta$.
    Let a point $p$ on $P$ be matched to a point $q$ on $Q$ by $(f, g)$, and let $p'$ and $q'$ be the respective points after straightening.
    Without loss of generality, assume that $p \in \CP_i \circ \BP_i$ for some $i$ (that is, $\BP_i$ exists).
    By Theorem~\ref{thm:matching_restrictions}, we have that $q \in \CQ_i \circ \BQ_i \circ \CQ_{i+1}$.
    
    Let $I$ and $I'$ be the regions used for straightening $\CP_i \circ \BP_i$ and $\CQ_i \circ \BQ_i$.
    We first prove the statement for when $q \in \CQ_i \circ \BQ_i$.
    Note that if the $I$- and $I'$-straightenings happen in opposite directions, then the $I$-straightening happens in the direction of $I'$ and vice versa.
    So if $p$ and $q$ move in opposite directions during straightening, they move towards each other and we have $|p' - q'| \leq |p - q| \leq \delta$.

    We argue that if $p$ and $q$ move in the same direction, then $|p' - q'| \leq \delta$.
    Without loss of generality, assume that both the $I$- and $I'$-straightenings happen in the positive direction.
    Since $I \cap I' \neq \emptyset$, performing the two straightenings separately is equivalent to performing the $(I \cup I')$-straightening of the curves.
    Let $p = (\CP_i \circ \BP_i)(x)$ and $q = (\CQ_i \circ \BQ_i)(y)$.
    We have that $p' = \min\{ \max_{x' \leq x} (\CP_i \circ \BP_i)(x'), \CP_{i+1}(0) \}$ and $q' = \min\{ \max_{y' \leq y} (\CQ_i \circ \BQ_i)(y'), \CP_{i+1}(0) \}$.
    If $|p' - q'| > \delta$, then there is either a point on $(\CP_i \circ \BP_i)[0, x]$ that is not close to any point on $(\CQ_i \circ \BQ_i)[0, y]$, or vice versa.
    Our construction of the subcurves gives us that all points on $\CP_i$ are close to points on $\CQ_i$ and vice versa, and symmetrically for $\BP_i$ and $\BQ_i$.
    Thus we have a contradiction, from which it follows that $|p' - q'| \leq \delta$ when $q \in \CQ_i \circ \BQ_i$.
    
    Now assume that $q \in \CQ_{i+1}$.
    Theorem~\ref{thm:matching_restrictions} then gives us that $p \in \BP_i$.
    Note that by the same theorem, there must now be a point $p'' \in \BP_i$ before $p$, that is matched to $\CQ_{i+1}(0)$.
    We again assume without loss of generality that $\CP_{i+1}(0) > \CP_i(1)$.
    This means that $\BP_i$ is $2\delta$-direction-preserving in the positive direction, which implies that $p \geq p'' - 2\delta \geq \CQ_{i+1}(0) - 3\delta \geq \CP_{i+1}(0) - 4\delta$.
    By construction we also have $p \leq \CP_{i+1}(0)$, and thus $p \in I'$.
    
    Let $p = \BP'_i(x)$ and $q = \CQ_{i+1}(y)$.
    Let $I''$ be the interval $I'$ expanded on both sides by $\delta$.
    Then $p' = \min\{ \max_{x' \leq x} \BP'_i(x'), \CP_{i+1}(0) \}$ and $q' = \min\{ \max_{y' \leq y} \CQ_{i+1}(y'), \CP_{i+2}(0) - 4\delta \} \geq \min\{ \max_{y' \leq y} \CQ_{i+1}(y'), \CP_{i+1}(0) \}$.
    Like before, if $|p' - q'| > \delta$, then there is either a point on $\BP'_i[0, x] \cap I'$ that is not close to any point on $\CQ_{i+1}[0, y] \subseteq \CQ_{i+1} \cap I''$, or vice versa.
    Note that $\BP'_i[0, x] \cap I' = \BP_i[0, x] \cap I'$, as $\BP'_i$ is the $I$-straightening of $\BP_i$, where $I \cap I' = \emptyset$.
    By our construction, we have that all points in $\BP_i \cap I'$ are close to points on $\CP_{i+1} \cap I''$ and vice versa.
    This gives a contradiction, from which it follows that $|p' - q'| \leq \delta$ when $q \in \CQ_{i+1}$.
\end{proof}

\begin{lemma}
\label{lem:distances_preserved_backwards}
    If $d_F(P^*, Q^*[0, y^*]) \leq \delta$ for some $y^* \in [0, 1]$, then $d_F(P, Q[0, y^*]) \leq 4L$.
\end{lemma}
\begin{proof}
    We prove that $|P(x) - P^*(x)| \leq 4L - 2\delta$, and similarly that $|Q(y) - Q^*(y)| \leq 4L$, for any $x \in [0, 1]$ and $y \in [0, t]$.
    It then follows that any $\delta$-matching between $P^*$ and $Q^*$ is also an $4L$-matching between $P$ and $Q$.
    
    Let $x \in [0, 1]$, and assume without loss of generality that $P(x) \in \CP_i \circ \BP_i$ for some $i$.
    Let $I$ be the interval such that $\CP_i'$ and $\BP_i'$ are the $I$-straightenings of $\CP_i$ and $\BP_i$, respectively.
    Also, let $I'$ be the interval such that $\BP_i''$ is the $I$-straightening of $\BP_i'$.
    
    Without loss of generality, assume that the $I$-straightenings happen in the positive direction.
    First assume that $P(x) = \CP_i(x')$ for some $x' \in [0, 1]$.
    Then because $\CP_i \subseteq I$, we have that $P^*(x) = \CP_i'(x') = \max_{x'' \leq x} \CP_i(x'')$.
    By Lemma~\ref{lem:compact_properties}, the diameter of $\CP_i$ is at most $4L - 2\delta$.
    It follows that $|P(x) - P^*(x)| \leq 4L - 2\delta$ if $P(x) \in \CP_i$.
    
    Now assume that $P(x) = \BP_i(x')$ for some $x' \in (0, 1)$.
    Note that because the $I$-straightening happens in the positive direction, we have that $\BP_i$ is $2\delta$-direction-preserving in the positive direction.
    Thus, if $P(x) \in I$, we have that $\BP_i'(x') = \max_{x'' \leq x} \BP_i(x) \geq \BP_i(x) - 2\delta$, and therefore that $|P(x) - P^*(x)| \leq 2\delta < 4L - 2\delta$.
    If $P(x) \in I'$ instead, then because $I'$ has diameter $4\delta$, it follows that $|P(x) - P^*(x)| \leq 4\delta < 4L - 2\delta$.
    We therefore have that $|P(x) - P^*(x)| \leq 4L - 2\delta$ if $P(x) \in \BP_i$.
    
    Now, to see that $|Q(y) - Q^*(y)| \leq 4L$, note that because the diameter of $\CP_i$ is at most $4L - 2\delta$, our construction leads to $\CQ_i$ having a diameter of at most $4L$.
    Also, recall that we assumed $\BQ_i$ to be $4\delta$-direction-preserving, as otherwise $d_F(P, Q) > \delta$ and we terminate the algorithm immediately.
    Using these two properties, we can apply the above proof for $|P(x) - P^*(x)|$ here as well.
\end{proof}

From Lemma~\ref{lem:distances_preserved} we obtain that any $\delta$-reachable point in $\F_{\leq \delta}(P, Q)$ is also a $\delta$-reachable point in $\F_{\leq \delta}(P^*, Q^*)$, and hence must be contained in any exit set.
Lemma~\ref{lem:distances_preserved_backwards} gives us that any $c\delta$-reachable point in $\F_{\leq \delta}(P^*, Q^*)$ is also $4cL$-reachable in $\F_{\leq \delta}(P, Q)$.
The above directly implies the following theorem:

\begin{lemma}
\label{lem:exit_sets}
    Let $P, Q \from [0, 1] \to \R^d$ be two curves with $m$ and $n$ vertices, respectively.
    Let $\delta > 0$ and $\alpha \in [1, n]$, and let $c \geq 1$ be a constant.
    In $O(m + n)$ time, we can construct two curves $P^*$ and $Q^*$, where $P^*$ is quasi-piecewise $(>4\delta)$-monotone, such that a $(c, \delta)$-exit set with respect to $P^*$ and $Q^*$ is an $(O(L / \delta), \delta)$-exit set with respect to $P$ and $Q$.
\end{lemma}

Together with the algorithm of Section~\ref{sec:piecewise_monotone} for computing exit sets in linear time if one of the curve is quasi-piecewise $(>4\delta)$-monotone, the above lemma immediately implies an algorithm for computing an $(O(L / \delta), \delta)$-exit set with respect to $P$ and $Q$ in $O(m + n)$ time.
This algorithm is used in the black box procedure of Theorem~\ref{thm:black_box_decision} in Section~\ref{sec:algorithm} to obtain an approximate decision algorithm for the \f distance.

\bibliographystyle{plainurl}
\bibliography{bibliography}

\end{document}